\documentclass[11pt]{article}

\usepackage[utf8]{inputenc}
\usepackage{libertine}
\usepackage{amsmath, amssymb, amsthm}
\usepackage{dsfont}
\usepackage{geometry}
\usepackage{hyperref}
\usepackage{comment}
\usepackage{color}
\usepackage[capitalize, nameinlink]{cleveref}
\usepackage[style=alphabetic]{biblatex}
\usepackage{enumitem}
\usepackage{graphicx}
\usepackage{thm-restate}
\usepackage{enumitem} 
\usepackage{caption}
\usepackage{extarrows}
\usepackage{xparse} 
\usepackage[most]{tcolorbox}
\usepackage[xcolor,hyperref,notion,electronic,makeidx,]{knowledge}
\usepackage{tikz}     
\usetikzlibrary{decorations.pathreplacing, arrows.meta}

\definecolor{tealblue}{HTML}{0052A3}

\knowledgestyle{notion}{color=tealblue}
\knowledgestyle{intro notion}{color=tealblue,italic}
\knowledge{notion}
    | Promise Constraint Satisfaction Problem
    | Promise Constraint Satisfaction Problems
    | PCSP
\knowledge{notion}
    | quantitative
\knowledge{notion}
    | qualitative
\knowledge{notion}
    | Approximate Graph Coloring
    | Approximate Coloring

\knowledge{notion}
    | relational signature
    | signature
\knowledge{notion}
    | relational symbols
    | symbols
    | symbol
\knowledge{notion}
    | Relational structure
    | relational structure
    | relational structures
    | Structure
    | structure
    | structures
    | Structures
\knowledge{notion}
    | Universe
    | universe
    | universes
\knowledge{notion}
    | similar
\knowledge{notion}
    | homomorphic map
    | homomorphism
    | homomorphisms
    | homomorphically
\knowledge{notion}
    | PCSP template
    | template
    | Template

\knowledge{notion}
    | biased distribution
    | Biased distribution
    | biased distributions

\knowledge{notion}
    | CSP
    | Constraint Satisfaction Problem
    | Constraint Satisfaction Problems
    | constraint satisfaction problem
\knowledge{notion}
    | CSP Dichotomy Conjecture
\knowledge{notion}
    | Polymorphism
    | polymorphism
    | polymorphisms
    | Polymorphisms
\knowledge{notion}
    | minor
    | Minor
    | minors
    | Minors
\knowledge{notion}
    | minor map
    | Minor map
    | minor maps
    | Minor maps
\knowledge{notion}
    | minion
    | Minion
    | Minions
    | minions
\knowledge{notion}
    | projection
    | Projection
    | projections
    | Projections
\knowledge{notion}
    | minion homomorphism
\knowledge{notion}
    | choice function
    | choice functions
\knowledge{notion}
    | Single choice condition
    | single choice condition
\knowledge{notion}
    | CSP Dichotomy Theorem
\knowledge{notion}
    | Multiple choice condition
    | multiple choice condition
\knowledge{notion}
    | Label Cover
    | label cover
\knowledge{notion}
    | constraint
    | constraints
    | Constraints
    | Constraint
\knowledge{notion}
    | labeling
    | Labeling
    | labelings
\knowledge{notion}
    | Gap Label Cover
    | gap label cover
\knowledge{notion}  
    | perfect completeness
    | Perfect Completeness
\knowledge{notion}
    | Layered choice condition
    | layered choice condition
\knowledge{notion}
    | Injective layered condition
    | injective layered condition
\knowledge{notion}
    | Unique Games Conjecture
\knowledge{notion}
    | 2-to-1
    | 2-to-1 functions
\knowledge{notion}
    | rich
    | Rich
    | richness
    | Richness
\knowledge{notion}
    | Rich 2-to-1 Label Cover
    | Rich 2-to-1 Gap Label Cover
\knowledge{notion}
    | Rich 2-to-1 Conjecture
\knowledge{notion}
    | Random 2-to-1 condition
    | random 2-to-1 condition
\knowledge{notion}
    | symmetric
    | Symmetric
    | symmetric function
    | Symmetric function
    | Symmetric functions
    | symmetric functions
\knowledge{notion}
    | monotone
    | Monotone
\knowledge{notion}
    | threshold function
    | threshold functions
    | threshold
    | thresholds
    | Thresholds
    | Threshold functions
    | Threshold function
    
\knowledge{notion}
    | Fourier analysis of Boolean functions
    | Fourier analysis
\knowledge{notion}
    | Influence
    | influence
    | influences
\knowledge{notion}
    | total influence
    | average sensitivity
\knowledge{notion}
    | Linear Threshold Function
    | Linear Threshold Functions
\knowledge{notion}
    | Polynomial Threshold Functions
    | Polynomial Thresholds
    | polynomial threshold functions
    | Polynomial Threshold Function

\knowledge{notion}
    | identity
    | Identity
    | identities
    | Identities
\knowledge{notion}
    | Minor condition
    | minor condition
    | minor conditions
    | Minor conditions
\knowledge{notion}
    | satisfied
    | Satisfied
\knowledge{notion}
    | symbol interpretation
    | Symbol interpretation
    | interpretation
    | Interpretation
\knowledge{notion}
    | Trivial
    | trivial
\knowledge{notion}
    | Promise Minor Condition
\knowledge{notion}
    | completeness
\knowledge{notion}
    | soundness

\knowledge{notion}
    | unbiased
\knowledge{notion}
    | positive
\knowledge{notion}
    | normalized
\knowledge{notion}
    | positive representation
    | representation
    | representations
\knowledge{notion}
    | Positive Polynomial Threshold Function
    | Positive Polynomial Threshold Functions
    | positive Polynomial Threshold Function
    | positive polynomial threshold function
    | Positive Polynomial Threshold
    | positive polynomial threshold
    | Positive Polynomial Thresholds
\knowledge{notion}
    | Weight of coordinate
    | weight of coordinate
    | coordinate weight
    | weight
    | weights
\knowledge{notion}
    | Regularity condition
    | regularity condition
    | regularity
\knowledge{notion}
    | Shapley Values
    | Shapley values
    | Shapley Value
    | Shapley value
\knowledge{notion}
    | McDiarmid's Inequality
    | McDiarmid's inequality
\knowledge{notion} 
    | bounded differences property
\knowledge{notion}
    | heavy sets
    | Heavy sets
    | heavy set
    | Heavy set

\knowledge{notion}
    | Fourier characters
    | Fourier character
\knowledge{notion}
    | Fourier decomposition
    | Fourier expansion
\knowledge{notion}
    | Fourier coefficients
\knowledge{notion}
    | Parseval-Plancherel identity
    | Parseval-Plancherel
    | Parseval identity
    | Plancherel identity
    | Parseval
    | Plancherel
\knowledge{notion}
    | degree truncations
    | degree truncation
    | truncation
    | truncations
    | high-degree part
    | high-degree truncation
    | high-degree
    | low-degree part
    | low-degree truncation
    | low-degree
\knowledge{notion}
    | pull-back distribution
    | Pull-back distribution
    | pull-back
    | Pull-back
\knowledge{notion}
    | low-degree concentration
    | low-degree concentrated
\knowledge{notion}
    | noisy distribution
    | Noisy distribution
\knowledge{notion}
    | Noise operator
    | noise
    | noise operator
    | Noise
\knowledge{notion}
    | Noise sensitivity
    | noise sensitivity

\graphicspath{ {.} }

\hypersetup{
    colorlinks=true,        
    linkcolor=tealblue,     
    citecolor=magenta,      
    urlcolor=magenta,       
    pdfborder={0 0 0},      
}

\newtheorem{theorem}{Theorem}[section]
\newtheorem{lemma}[theorem]{Lemma}
\newtheorem{corollary}[theorem]{Corollary}
\newtheorem{proposition}[theorem]{Proposition}
\theoremstyle{definition}
\newtheorem{definition}[theorem]{Definition}
\newtheorem{example}[theorem]{Example}
\theoremstyle{remark}

\theoremstyle{plain}
\newtheorem{claim}[theorem]{Claim}
\theoremstyle{plain}
\newtheorem{condition}[theorem]{Condition}
\theoremstyle{plain}
\newtheorem{question}[theorem]{Question}
\theoremstyle{plain}
\newtheorem{conjecture}[theorem]{Conjecture}

\title{On Boolean PCSPs with Polynomial Threshold Polymorphisms}
\author{
{\color{tealblue} Katzper Michno} \\
\small katzper.michno@gmail.com \\
\small Faculty of Mathematics and Computer Science \\
\small Jagiellonian University \\
\small Kraków, Poland
}
\date{}

\newcommand{\PCSP}{\mathsf{PCSP}}
\newcommand{\CSP}{\mathsf{CSP}}
\newcommand{\EX}{\mathop{\mathds{E}}}
\newcommand{\norm}[1]{\left\lVert #1 \right\rVert}

\newcommand{\T}{\mathbf{T}}

\newcommand{\NS}{\mathbf{NS}}

\newcommand{\minion}{\mathcal{M}}
\newcommand{\Pol}{\mathsf{Pol}}

\newcommand{\Deg}[2]{\ensuremath{\omega_{#2}[#1]}}
\newcommand{\NP}{\mathsf{NP}}
\newcommand{\PTF}{\mathsf{PTF}}
\newcommand{\LTF}{\mathsf{LTF}}

\newcommand{\PMC}{\mathsf{PMC}}
\newcommand{\THR}{\mathsf{THR}}
\newcommand{\repr}[2]{\ensuremath{\langle #1 |^{#2}}}
\newcommand{\tuple}[1]{\overline{#1}}
\newcommand{\Bool}{\ensuremath{\{0,1\}}}
\newcommand{\A}{\ensuremath{\mathbb{A}}}
\newcommand{\B}{\ensuremath{\mathbb{B}}}
\newcommand{\X}{\ensuremath{\mathbb{X}}}
\newcommand{\GapRich}{\ensuremath{\mathsf{Gap\text{-}Rich\text{-}2\text{-}to\text{-}1}}}
\newcommand{\Gap}{\ensuremath{\mathsf{Gap\text{-}Label\text{-}Cover}}}
\newcommand{\GapUnique}{\ensuremath{\mathsf{Gap\text{-}Unique}}}
\newcommand{\GapDtoOne}[1]{\ensuremath{\mathsf{Gap\text{-}#1\text{-}to\text{-}1}}}

\NewDocumentCommand{\cube}{>{\SplitArgument{1}{,}}m}{%
  \cubeaux#1%
}
\NewDocumentCommand{\cubeaux}{mm}{%
  \IfNoValueTF{#2}{\ensuremath{\mu_{#1}}}{\ensuremath{\mu_{#1}^{(#2)}}}%
}

\NewDocumentCommand{\pullback}{>{\SplitArgument{1}{,}}m}{%
  \pullbackaux#1%
}
\NewDocumentCommand{\pullbackaux}{mm}{%
  \IfNoValueTF{#2}{\ensuremath{\nu_{#1}}}{\ensuremath{\nu_{#1}^{(#2)}}}%
}

\NewDocumentCommand{\Inf}{o >{\SplitArgument{1}{,}}m}{%
  \IfNoValueTF{#1}
    {\InfProcessTwoArgs#2}      
    {\InfProcessThreeArgs{#1}#2} 
}
\NewDocumentCommand{\InfProcessTwoArgs}{mm}{%
  \mathbf{Inf}{#2}[#1]
}
\NewDocumentCommand{\InfProcessThreeArgs}{mmm}{%
  \mathbf{Inf}\sp{#1}\sb{#3}[#2]
}
 
\NewDocumentCommand{\I}{o m}{%
  \ensuremath{%
    \mathbf{I}%
    \IfValueT{#1}{^{#1}}
    [#2]
  }%
}

\DeclareDelimFormat{multicitedelim}{\addcomma\space} 
\AtEveryBibitem{
  \clearfield{url}
  \clearfield{eprint}
}
\DeclareFieldFormat{doi}{
  DOI:\space
  \addcolon\space\href{https://doi.org/#1}{\nolinkurl{#1}}%
}
\begin{document}

\maketitle

\begin{abstract}
    In pursuit of a deeper understanding of Boolean Promise Constraint Satisfaction Problems ($\PCSP$s), we identify a class of problems with restricted structural complexity, which could serve as a promising candidate for complete characterization. Specifically, we investigate the class of $\PCSP$s whose polymorphisms are Polynomial Threshold Functions ($\PTF$s) of bounded degree. We obtain two complexity characterization results: (1) with a hardness condition introduced in \cite{layers}, we establish a complete complexity dichotomy in the case where coefficients of $\PTF$ representations are non-negative; (2) dropping the non-negativity assumption, we show a hardness result for $\PTF$s admitting coordinates with significant influence, conditioned on the Rich 2-to-1 Conjecture proposed in \cite{braverman_et_al:LIPIcs.ITCS.2021.27}. In order to prove the latter, we show that a random 2-to-1 minor map retains significant coordinate influence over the $p$-biased hypercube with constant probability.
\end{abstract}

\section{Introduction}\label{sec:introduction}
The class of \intro{Promise Constraint Satisfaction Problems} ($\PCSP$) can be seen as \textit{qualitative} approximation problems within $\NP$, originally introduced in \cite{hastad_guruswami_2017}. Perhaps the most famous $\PCSP$ problem is the \intro{Approximate Graph Coloring}, which for a pair of parameters $s < t$ asks to find a $t$-coloring of the input graph, which is promised to be $s$-colorable. The question of classifying the complexity of this problem is known for its notorious difficulty. Although polynomial-time algorithms that solve the case where $t$ is a sublinear function of the number of vertices have been known for decades and are still improving, currently achieving $t \approx n^{0.2}$ for $3$-colorable graphs  \cite{wigderson_1983, blum_94, kawarabayashi_2017, kawarabayashi_2024}, not much is known about the case where $t$ is a constant. We know that if $s = 3$ and $t = 5$, then the problem is $\NP$-hard \cite{Guruswami_2000, Safra_2000}. For general $s$, the strongest known result states that \kl{Approximate Coloring} is $\NP$-hard for $t = 2s-1$. In particular, the case for $s = 3, t = 6$ remains open. However, we know that the problem is hard for all constants $s,t$ under the assumption of a certain variant of Khot's \kl{Unique Games Conjecture} \cite{dinur_approximate_graph_coloring, dTo1Hardness}.

\AP
The \kl{Approximate Coloring} problem is only a single representative of the extremely rich class of $\PCSP$s. We generally define $\PCSP$s using the language of \kl{relational structures}, which we now briefly introduce. A \intro{relational signature} consists of a sequence of \intro{relational symbols} $R_1, \dots, R_k$ with fixed corresponding arities $r_1, \dots, r_k$. A \intro{relational structure} over \intro{universe} $A$ of such a signature is a tuple $\mathbb{A} = (A; R_1^{\mathbb{A}}, \dots, R_k^{\mathbb{A}})$, where for every $i \in [k]$, $R_i^{\mathbb{A}} \subseteq A^{r_i}$ is a relation of arity corresponding to the arity of the symbol $R_i$. We say that two relational structures are \intro{similar} if they have the same signature. Given two similar structures $\A$ and $\B$ over universes $A$ and $B$ respectively, by $\A \to \B$ we denote the fact that there exists a \intro{homomorphic map} from $\mathbb{A}$ to $\mathbb{B}$, i.e. a function $\varphi : A \to B$ that preserves relations. Formally, for every symbol $R$ of arity $n$ in the signature, $\varphi$ must satisfy that
\[
    \forall \, a_1, \dots, a_n \in A : (a_1, a_2, \dots, a_n) \in R^\A \implies (\varphi(a_1), \varphi(a_2), \dots, \varphi(a_n)) \in R^\B.
\]

To define a $\PCSP$ problem, we need two similar structures $\A, \B$ such that $\A \to \B$. We say that such a pair of structures $(\A, \B)$ is a $\PCSP$ \intro{template}. Now, the problem $\PCSP(\A, \B)$ can be stated as follows: Given an input structure $\X$ such that $\X \to \A$, find and output a \kl{homomorphism} from $\X$ to $\B$. The assumption that $\X \to \A$ is the \textit{promise}, which is not part of the input, but guarantees that a solution exists, because $\A \to \B$ and homomorphisms are composable. Observe that in this framework, the \kl{Approximate Coloring} problem is exactly $\PCSP(K_s, K_t)$\footnote{Here $K_t$ is a complete undirected graph on $t$ vertices, i.e. a relational structure with universe of size $t$ and one binary relation consisting of all pairs of distinct elements.}. 

The above definition yields the \textit{search} variant of $\PCSP$s. We also consider the \textit{decision} version: In this variant, the problem $\PCSP(\A, \B)$ asks us to distinguish between the cases where $\X \to \A$ and $\X \not \to \B$. Here, the requirement $\A \to \B$ ensures that these two scenarios are disjoint and that the problem is well defined. It is clear that every $\PCSP$ is in $\NP$ (as long as the universes of the underlying structures are finite) and that the decision variant trivially reduces to the search version. However, it is unknown whether the opposite statement is true.

We can see that the framework of $\PCSP$s is highly expressive and only a small portion of $\PCSP$ problems have simple combinatorial interpretations like \kl{Approximate Coloring}. It is therefore imperative to search for abstract tools in order to successfully classify larger chunks of $\PCSP$ problems. The general, state-of-the-art machinery for this purpose is the \textit{algebraic approach}. To properly introduce the algebraic approach, we have to take a step back and talk about the class of \kl{Constraint Satisfaction Problems} --- the predecessors of $\PCSP$s.

\subsection{Constraint Satisfaction Problems and the algebraic approach}

The class of \intro{Constraint Satisfaction Problems} ($\CSP$) has been present in the theory of Computational Complexity for decades. The historically original formulation of $\CSP$ problems is as follows: Every $\CSP$ consists of a domain, variables, and constraints, and the goal is to find an evaluation of the variables that satisfies all given constraints. The classic example of a $\CSP$ is $\mathsf{3}$-$\mathsf{SAT}$ --- given a set of variables and clauses, we want to find an evaluation of variables that satisfies all clauses, which play the role of constraints. 

Although the above definition explains the name of the $\CSP$ class, we will use another, equivalent formulation of $\CSP$s in the language of \kl{relational structures}. Given a structure $\A$, the problem $\CSP(\A)$ in the search variant asks us to find a \kl{homomorphism} of the input structure $\X$ to $\A$. Similarly to $\PCSP$s, the decision version of a $\CSP$ only requires to decide whether $\X \to \A$. Interestingly, it is known that every search $\CSP$ reduces to its decision version \cite{polymorphisms_and_how_to_use_them}.

\begin{example}[$\mathsf{3}$-$\mathsf{SAT}$]\label{example:3-SAT}
    Consider a relational structure $\mathbb{A} = (\{0,1 \}; \{0, 1\}^3 - \{(0,0,0)\}, \neq)$, which consists of a Boolean universe and two relations. The first relation is ternary and consists of all Boolean triples with at least one $1$. The second relation is the binary inequality $\{(0,1), (1,0)\}$. Every instance of $\mathsf{3}$-$\mathsf{SAT}$ with a set of literals $X = \{x_1, \dots, x_n, \neg x_1, \dots, \neg x_n \}$ and clauses $C \subseteq X^3$ can be encoded as a relational structure $\mathbb{X} = (X; C, \{(x_i, \neg x_i) : i \in [n] \})$. We can see that the formula is satisfiable if and only if $\mathbb{X} \to \mathbb{A}$, and thus $\mathsf{3}$-$\mathsf{SAT}$ is exactly the problem $\CSP(\mathbb{A})$. 
\end{example}

\begin{example}[$\mathsf{3}$-$\mathsf{COLORING}$]\label{example:3-colouring}
    The problem of coloring a given graph with 3 colors can be equivalently stated as the problem of finding a homomorphism of the input graph into a triangle $K_3$. The vertices of $K_3$ serve as colors, and since $K_3$ has no loops, any pair of vertices connected by an edge must be assigned to different vertices in $K_3$. Therefore, the problem is equivalent to $\CSP(K_3)$.
\end{example}

The systematic study of $\CSP$s has already begun in the previous century. One of the first works in this area is due to Thomas Schaefer \cite{schaefer} and his \textit{dichotomy} result, which states that every $\CSP$ defined by a structure with Boolean universe $\{0,1\}$ is $\NP$-hard or solvable in polynomial time. The next milestone was established by Hell and Nešetřil \cite{HELL199092}, who showed that an analogous dichotomy is true for graph $\CSP$s, i.e. such that the underlying structure has only one relation, which is binary and symmetric. These two results motivated a conjecture proposed in 1998 by Feder and Vardi \cite{feder_vardi_csp_dichotomy_conjecture}, which stated that the entire class of $\CSP$s with finite domain admits such dichotomy. This problem is today known by the name of \intro{CSP Dichotomy Conjecture} and was the main motivation of numerous major breakthroughs in the study of $\CSP$s for almost three decades, until its positive resolution by Zhuk and Bulatov in 2017 \cite{zhuk_dichotomy_conf, bulatov_dichotomy, zhuk_dichotomy}. 

To understand the methodology of studying $\CSP$s that led to this celebrated conclusion, we first have to introduce relevant notions. Typically, we write tuples with a bar above, e.g. $\tuple x, \tuple y$. If a tuple $\tuple x$ is indexed by $[n]$, then $x_i$ is the value of $i$-th element of $\tuple x$ for every $i \in [n]$. Perhaps the most crucial concept in the whole of $\CSP$ research are \textit{polymorphisms} of relational structures.

\begin{definition}[\intro{Polymorphism}]
    Suppose that $\A$ and $\B$ are two \kl{similar} structures of signature $\tau$ over universes $A$ and $B$, respectively. We say that a function $\varphi : A^n \to B$ is an $n$-ary \textit{polymorphism} of $(\A, \B)$, if given any symbol $R$ in $\tau$ of arity $m$ and tuples $\tuple a^{(1)}, \dots, \tuple a^{(n)} \in R^\A$, we have
    \[
        \big(\varphi(a_1^{(1)}, \dots, a_1^{(n)}), \dots, \varphi(a_m^{(1)}, \dots, a_m^{(n)})\big) \in R^\B.
    \]
    A convenient way to think about polymorphisms is as follows: Suppose that we have a $m \times n$ matrix consisting of columns $\tuple a^{(1)}, \dots, \tuple a^{(n)} \in R^\A$. If we apply a polymorphism $\varphi$ to rows of the matrix and obtain values $b_1, \dots, b_m$, then it must hold that $(b_1, \dots, b_m) \in R^\B$, as the following picture depicts:
    \[
    \renewcommand{\arraystretch}{1.5} 
    \large 
    \begin{array}{c@{}c@{\hspace{1.8mm}}c@{\hspace{1.8mm}}c@{}c@{}} 
    \begin{tikzpicture}[baseline=-0.6ex]
        \node at (-5,0) {$\big( R^\A \big)^n \ni$};
    \end{tikzpicture}
    &
    \begin{bmatrix}
    a_1^{(1)} & a_1^{(2)} & \cdots & a_1^{(n)} \\
    a_2^{(1)} & a_2^{(2)} & \cdots & a_2^{(n)} \\
    \vdots    & \vdots    & \ddots & \vdots    \\
    a_m^{(1)} & a_m^{(2)} & \cdots & a_m^{(n)} \\
    \end{bmatrix}
    &
    \begin{tikzpicture}[baseline=-2ex]
        \foreach \i in {1,...,4} {
            \draw[-Latex, thick] (-0.2,-\i*0.73+1.5) -- (1.5,-\i*0.73+1.5) 
                node[midway, above] {$\varphi$};
        }
    \end{tikzpicture}
    &
    \begin{bmatrix}
    b_1 \\ b_2 \\ \vdots \\ b_m
    \end{bmatrix}
    &
    \begin{tikzpicture}[baseline=-0.6ex]
        \node at (-5,0) {$\in R^\B$};
    \end{tikzpicture}
    \end{array}
    \]    
\end{definition}

Alternatively, polymorphisms can be defined as \kl{homomorphisms} of the form $\A^n \to \B$, where the power of relational structure is defined in a standard way. We denote the set of polymorphisms of $(\A, \B)$ by $\Pol(\A, \B)$. In the special case where $\A = \B$, we write $\Pol(\A)$ and simply call it the set of structure polymorphisms. Arguably, the most important polymorphisms are \intro{projections}, i.e. functions of the form $(x_1, \dots, x_n) \mapsto x_i$. Indeed, it is an easy exercise to show that all projections of all arities over the universe of $\A$ belong to $\Pol(\A)$ (see \cref{fig:K_3_polymorphism} for an example of a projection as a polymorphism).

\begin{figure}[hbtp]
    \centering
    \includegraphics[scale=0.85]{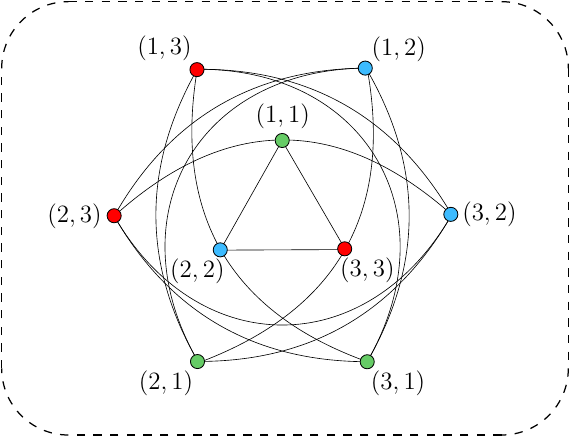}
    \caption{
        $K_3^2$ is a structure consisting of universe $[3]^2$ and a single symmetric, binary relation $\leftrightarrow$ defined as: $(a,b) \leftrightarrow (c,d)$ if and only if $a \leftrightarrow b$ and $c \leftrightarrow d$. It is simple to verify that the function $f: (x, y) \mapsto y$ is a \kl{homomorphism} of $K_3^2$ to $K_3$. As a consequence, $f$ is a \kl{polymorphism} of $K_3$.
    }\label{fig:K_3_polymorphism}
\end{figure}

The first significant progress in general understanding of $\CSP$s was the invention of the \textit{algebraic approach} \cite{algebraic_approach_1, algebraic_approach_2, algebraic_approach_3}. It established a tight connection between the complexity of $\CSP(\A)$ and the set $\Pol(\A)$. What was quickly observed is that the richer the set of polymorphisms is, the easier the $\CSP$ gets, in the sense that if $\Pol(\A) \subseteq \Pol(\B)$, then $\CSP(\B)$ reduces to $\CSP(\A)$ in logarithmic space. Therefore, if \kl{CSP Dichotomy Conjecture} was supposed to be true, the question remained: How \textit{exactly} poor must the set $\Pol(\A)$ be for $\CSP(\A)$ to be $\NP$-hard? To answer this question, we have to understand \kl{minors}, which are the next central notion in the algebraic approach to $\CSP$s.

\begin{definition}[\intro{Minor}]
    Suppose that $f : A^n \to B$ is a function and $\pi : [n] \to [m]$ is a map. By $f \xrightarrow[]{\pi} g$ or $g = f^\pi$, we denote that $g$ is an $m$-ary \textit{minor} of $f$ with respect to \intro{minor map} $\pi$, which means that for every choice of $a_1, \dots, a_m \in A$, we have
    \[
        g(a_1, a_2, \dots, a_m) = f\big(a_{\pi(1)}, a_{\pi(2)}, \dots, a_{\pi(n)}\big).
    \]
\end{definition}

In other words, a minor is a function that can be obtained from the original function by permuting, adding, and identifying variables. We are particularly interested in sets of functions that are closed under the operation of taking minors --- such sets are called \intro{minions}. Crucially, $\Pol(\A, \B)$ is always a minion, i.e. every minor of a polymorphism is also a polymorphism of the same structures. Similarly, the set of all \kl{projections} over a fixed set is a minion. In particular, we distinguish $\mathcal{P}_2$ --- the set of all projections over a set of size 2. 

\AP
In 2005, Bulatov, Jeavons and Krokhin \cite{algebraic_approach_1} proposed a universal hardness condition for $\CSP$s. They showed that if $\Pol(\A)$ resembles $\mathcal{P}_2$ in a certain way, then $\CSP(\A)$ must be $\NP$-hard. To formalize their condition, we consider the notion of \intro{choice functions}: Given any \kl{minion} $\minion$, we say that $C : \minion \to 2^{\mathbb{N}}$ is a choice function for $\minion$ if for every function $f \in \minion$ of arity $n$, we have $C(f) \subseteq [n]$. In other words, a choice function is a selection of some distinguished coordinates in every function of $\minion$. We are mainly interested in choice functions that are \textit{compatible} with \kl{minors}, i.e. such that if $f \xrightarrow[]{\pi} g$, then $\pi(C(f)) \cap C(g) \neq \emptyset$. It is quite clear that such a choice exists for $\mathcal{P}_2$ --- in addition, it is sufficient to choose only one coordinate in every function. This is exactly what is called the \kl{single choice condition}, which serves as the formal notion of ``resemblance'' to $\mathcal{P}_2$.

\begin{condition}[\intro{Single choice condition}]
    A minion $\minion$ satisfies the \textit{single choice condition} if there exists a \kl{choice function} $C$ for $\minion$ such that for every $f \in \minion$:
    \begin{itemize}
        \item $|C(f)| = 1$, and
        \item if $f \xrightarrow[]{\pi} g$, then $\pi(C(f)) =~C(g)$.
    \end{itemize}
\end{condition}

What the authors of \cite{algebraic_approach_1} proved was that if $\Pol(\A)$ admits a so-called \textit{clone homomorphism}\footnote{The word \textit{clone} in this context is not coincidental --- it is a known fact that $\Pol(\A)$ is always closed under composition.} to $\mathcal{P}_2$, then $\CSP(\A)$ is $\NP$-hard. This condition was later weakened to admitting a \textit{minion homomorphism}\footnote{Minion homomorphisms were originally referred as \textit{h1 clone homomorphisms}, i.e. preserving all \textit{height 1} identitiies --- opposed to regular clone homomorphisms which are supposed to capture identities of all heights (with arbitrarily deep compositions).} to $\mathcal{P}_2$ \cite{wonderlands}, which can be equivalently presented as the \kl{single choice condition}. All that remained was to show the second part: that the problem is tractable in the opposite case. It took more than a decade to finish, but it was eventually achieved in 2017 independently by Zhuk \cite{zhuk_dichotomy_conf} and Bulatov \cite{bulatov_dichotomy} --- wrapping up the proof of the \kl{CSP Dichotomy Conjecture}.

\begin{theorem}[\intro{CSP Dichotomy Theorem}]
    Suppose that $\A$ is a \kl{relational structure} with finite \kl{universe}. Then
    \begin{enumerate}[label=(\arabic*)]
        \item either $\Pol(\A)$ satisfies the \kl{single choice condition} and $\CSP(\A)$ is $\NP$-hard, or
        \item $\Pol(\A)$ does not satisfy the \kl{single choice condition} and $\CSP(\A)$ is solvable in polynomial time.
    \end{enumerate}
\end{theorem}

\subsection{Adding the promise to CSPs}

The class of $\PCSP$s is clearly an extension of $\CSP$s: $\CSP(\A)$ can always be equivalently defined as $\PCSP(\A, \A)$. Furthermore, the algebraic approach to $\CSP$s quite neatly generalizes to $\PCSP$s. The complexity of $\PCSP(\A, \B)$ is fully captured by $\Pol(\A, \B)$: if $\Pol(\A, \B) \subseteq \Pol(\A', \B')$, then $\PCSP(\A', \B')$ log-space reduces to $\PCSP(\A, \B)$ \cite{algebraic_approach_to_pcsps}. This motivates the quest to extend the results in the field of $\CSP$s to the more general framework of $\PCSP$s. The following question seems to be a reasonable starting point:
\begin{center}
    \textit{Are analogues of Schaefer's and Hell-Nešetřil dichotomies true for Boolean and graph PCSPs?}
\end{center}
Unfortunately, it seems that we are still far from resolving this problem. As it turns out, $\PCSP$s are significantly more complex than $\CSP$s, even if the structure domains are restricted to $\{0,1\}$, or both are required to be undirected graphs. In this paper, we contribute to the former case by providing a new dichotomy result for Boolean $\PCSP$s with structurally limited \kl{polymorphisms}. Before we state our results, we explore the tools for studying $\PCSP$s developed in recent years.

One of the most apparent manifestations of $\PCSP$ complexity is that there is no apparent candidate for a universal hardness condition analogous to \kl{single choice condition}. In fact, even weaker conditions used for reductions in multiple partial hardness results have already been proven to be insufficient to capture the whole landscape of $\NP$-hard $\PCSP$s. The first results in the study of $\PCSP$s \cite{hastad_guruswami_2017, symmetric_boolean_dichotomy, symmetric_boolean_pcsps, 3_coloring_of_H_colorable_graphs} utilized the \textit{multiple choice condition}, which is a direct relaxation of the \kl{single choice condition}.

\begin{condition}[\intro{Multiple choice condition}]
    A minion $\minion$ satisfies the \textit{multiple choice condition} if there exists a number $M$ and a \kl{choice function} $C$ for $\minion$, such that for every $f \in \minion$:
    \begin{itemize}
        \item $|C(f)| \leq M$ for every $f \in \minion$, and 
        \item if $f \xrightarrow[]{\pi} g$, then $\pi(C(f)) \cap C(g) \neq \emptyset$.
    \end{itemize}
\end{condition}

The \kl{multiple choice condition} provides more flexibility in hardness results by allowing us to select more than one ``significant'' coordinate in every polymorphism. The fact that \kl{single choice condition} implies $\CSP$ hardness follows from a direct reduction from $\mathsf{3}$-$\mathsf{SAT}$. In contrast, the proof that $\PCSP(\A, \B)$ whose polymorphism minion satisfies the \kl{multiple choice condition} is $\NP$-hard is slightly more involved. Here, we have to use more sophisticated starting points for hardness reductions, which are various variants of the \textit{Label Cover} problem.

\begin{definition}[\intro{Label Cover}]
    An instance of \textit{Label Cover} $\Psi = (L \cup R, E, \Sigma_L, \Sigma_R, \Pi)$ consists of a bipartite graph $(L \cup R, E)$, alphabets $[\Sigma_L]$, $[\Sigma_R]$ and a set $\Pi$ of functional \intro{constraints} $\pi_e : [\Sigma_L] \to [\Sigma_R]$ for every $e \in E$. Given a \intro{labeling} $\sigma$, which assigns an element of $[\Sigma_L]$ to every vertex in $L$ and an element of $[\Sigma_R]$ to every vertex in $R$, we say that constraint $\pi_e \in \Pi$ for $e = (u, v) \in L \times R$ is satisfied if $\pi_e(\sigma(u)) = \sigma(v)$.
\end{definition}

It is $\NP$-hard to decide whether for a given instance of \kl{Label Cover}, a labeling satisfying all constraints exists. However, to obtain results regarding hardness of approximation (including $\PCSP$s), we consider the promise version of \kl{Label Cover}, called \textit{Gap Label Cover} problem.

\begin{definition}[\intro{Gap Label Cover}]
    Suppose that $n \geq 1$ and $s, t \in [0,1]$ are constants such that $s \leq t$. By $\Gap_n[t, s]$ we denote the decision problem defined as follows: Given a \kl{Label Cover} instance $\Psi$ with $\Sigma_L, \Sigma_R \leq n$, answer $\mathsf{YES}$ if there is a \kl{labeling} satisfying at least a $t$-fraction of \kl{constraints}, and answer $\mathsf{NO}$ if all labelings satisfy at most an $s$-fraction of constraints.
\end{definition}

The interesting aspect of \kl{Gap Label Cover} is that it is $\NP$-hard even in the \intro{perfect completeness} version, i.e. given that $t = 1$. Specifically, for every $s \in (0,1)$, the problem $\Gap_n[1, s]$ is $\NP$-hard (for sufficiently large $n$). This fact is an implication of the \textit{PCP Theorem} \cite{PCP_1, PCP_2} combined with \textit{Raz’s Parallel Repetition Theorem} \cite{parallel_repetition}, or alternatively, a combinatorial surrogate of PCP called \textit{Baby PCP Theorem} proposed in \cite{baby_pcp}. To show that \kl{multiple choice condition} implies hardness for $\PCSP$s, we have to reduce \kl{Gap Label Cover} with perfect completeness to a $\PCSP$. This is typically done with an intermediate problem called \textit{Promise Minor Condition}. We elaborate on this topic and explicitly describe such a reduction in \cref{sec:random_condition}.

Unfortunately, we already know that \kl{multiple choice condition} is too restrictive. In search of even weaker hardness conditions for $\PCSP$s, two main advances have been obtained. In 2021, Brands, Wrochna and Živný \cite{layers} proposed a \textit{layered} variant of \kl{Gap Label Cover} and showed that it is hard. This gave rise to the \textit{layered choice condition}, which has already proved to be useful in several new hardness results for $\PCSP$s \cite{layers, symmetric_pcsps_beyond_boolean, ciardo2023complexity}, where the vanilla \kl{multiple choice condition} was already lacking. 

\begin{condition}[\intro{Layered choice condition}]
    A minion $\minion$ satisfies the \textit{layered choice condition} if there exists a number $M$ and a \kl{choice function} $C$ for $\minion$, such that 
    \begin{itemize}
        \item $|C(f)| \leq M$ for every $f \in \minion$, and 
        \item for every $f_1 \in \minion$ and a chain of minor maps 
        \[
            f_1 \xrightarrow{\pi_{1,2}} f_2 \xrightarrow{\pi_{2,3}} \dots \xrightarrow{\pi_{M-1, M}} f_M,
        \] 
        there are $i < j$ satisfying $\pi_{i,j}(C(f_i)) \cap C(f_j) \neq \emptyset$ (where $\pi_{i, j} = \pi_{j-1, j} \circ \dots \circ \pi_{i, i+1}$ is a composition of minor maps).
    \end{itemize}
\end{condition}

\begin{theorem}[Corollary 4.2 in \cite{layers}]\label{theorem:reductions:layered_condition_implies_hardness}
    Let $(\mathbb{A}, \mathbb{B})$ be $\PCSP$ \kl{template}. If $\Pol(\mathbb{A}, \mathbb{B})$ satisfies the \kl{layered choice condition}, then $\PCSP(\mathbb{A}, \mathbb{B})$ (the decision version) is $\NP$-hard.
\end{theorem}

The most recent refinement in this line is by Banakh and Kozik \cite{injective_hardness_condition}, who identified an $\NP$-hard Boolean $\PCSP$ that does not satisfy the \kl{layered choice condition}. As a remedy, they proposed a \textit{smooth} version of \kl{Gap Label Cover} with layers and obtained the currently weakest $\PCSP$ hardness condition: the \intro{injective layered condition}, which requires the choice function to be compatible only with minor maps that are injective on the coordinates distinguished by the \kl{choice function}. This is the first $\PCSP$ hardness condition that imposes a \textit{structural} restriction on minor maps with which the choice function must be compatible. 

To obtain even stronger structural assertions in hardness conditions for $\PCSP$s, one has to introduce restrictions on \kl{constraints} in \kl{Label Cover} instances. In this paradigm, the famous \intro{Unique Games Conjecture} \cite{unique_games} asks whether the \kl{Gap Label Cover} problem remains $\NP$-hard if all constraints are bijective --- we denote this problem by $\GapUnique_n[t, s]$. It is straightforward to see that we cannot hope for perfect completeness as $\GapUnique[1, s]$ is solvable in polynomial time. However, if we allow the constraints to be \intro{2-to-1 functions}, i.e. such that $|\pi_e^{-1}(i)| = 2$, then the question of hardness with perfect completeness remains open. This problem is denoted by $\GapDtoOne{2}_n[t, s]$ and has seen significant traction in recent years. In a striking series of papers \cite{2TO1_1, 2TO1_side1, 2TO1_side2, 2TO1_2, 2TO1_3, 2TO1_4}, it has been shown that for every $\varepsilon > 0$, the problem $\GapDtoOne{2}_n[1-\varepsilon, \varepsilon]$ is $\NP$-hard for sufficiently large $n$. As a corollary, this implies the hardness of $\GapUnique_n[1/2, \varepsilon]$, reinforcing \kl{Unique Games Conjecture}. We note that hardness of $\GapDtoOne{2}_n[1, s]$ would have significant consequences for $\PCSP$s. In particular, it would imply hardness of \kl{Approximate Coloring} \cite{dinur_approximate_graph_coloring, dTo1Hardness}. 

More recently, Braverman, Khot and Minzer \cite{braverman_et_al:LIPIcs.ITCS.2021.27} proposed another variant of restriction on \kl{Label Cover} constraints, which they called \intro{richness}. We say that $\Psi$ is an instance of \intro{Rich 2-to-1 Label Cover}, if (1) all constraints are 2-to-1, and (2) the process of uniformly choosing a constraint adjacent to any fixed vertex in $L$ yields a uniform distribution over all 2-to-1 functions $[2n] \to [n]$\footnote{We note that in the original formulation of the conjecture, the richness condition is weaker: the distribution of paritions of $[2n]$ into pairs induced by pre-images of constraints is required to be uniform. Our formulation is more convenient for our considerations --- it comes from \cite{brakensiek_et_al:LIPIcs.ICALP.2021.37}.}. This variant of \kl{Gap Label Cover} is denoted by $\GapRich_n[t,s]$. Interestingly, its conjectured hardness without perfect completeness is equivalent to \kl{Unique Games Conjecture} \cite[Theorem 8]{braverman_et_al:LIPIcs.ITCS.2021.27}. This fact motivated the authors to propose the \textit{Rich 2-to-1 Conjecture}.

\begin{conjecture}[\intro{Rich 2-to-1 Conjecture}]
    For every $s \in (0,1)$, the problem $\GapRich_n[1, s]$ is $\NP$-hard for sufficiently large $n$.
\end{conjecture}

The application of \kl{Rich 2-to-1 Conjecture} in $\PCSP$s has been pioneered by Brakensiek, Guruswami and Sandeep \cite{brakensiek_et_al:LIPIcs.ICALP.2021.37}, who established a complete conditional dichotomy for Boolean $\PCSP$s with polymorphism minions consisting of \intro{monotone} functions, i.e. functions of the form $f : \Bool^n \to \Bool$ such that if $\tuple x \leq \tuple y$\footnote{Here the relation $\leq$ is coordinate-wise, i.e. $\tuple x \leq \tuple y$ if and only if $x_i \leq y_i$ for every coordinate $i$.}, then $f(\tuple x) \leq f(\tuple y)$. They showed that \kl{Rich 2-to-1 Label Cover} can be reduced to $\PCSP$s satisfying the \textit{random 2-to-1 condition}, which requires the choice function to be compatible only with randomly chosen 2-to-1 minor maps.

\begin{condition}[\intro{Random 2-to-1 condition}]
    Suppose that $\minion$ is a minion. We say that $\minion$ satisfies the \textit{random 2-to-1 condition} if there exist numbers $M, \tau > 0$ and a \kl{choice function} $\mathcal{C}$ for $\minion$, such that for every $f \in \minion$:
    \begin{enumerate}[label=(\arabic*)] 
        \item $|C(f)| \leq M$, and 
        \item if $f$ has even parity $\,2n$, then
        \[
            \Pr_{\pi} \big[ \pi(C(f)) \cap  C(f^\pi) \neq \emptyset \big] \geq \tau,
        \]
        where $\pi : [2n] \to [n]$ is a uniformly random $2$-to-$1$ minor map. 
    \end{enumerate}
\end{condition}

The \kl{random 2-to-1 condition} is highly flexible; we will use it for one of our conditional hardness results in \cref{sec:influence}. In \cref{sec:random_condition} we prove the following result, which has already been obtained in the setting of Boolean $\PCSP$s \cite[Theorem 4.8]{brakensiek_et_al:LIPIcs.ICALP.2021.37}. 

\begin{theorem}\label{theorem:reductions:random_condition_implies_hardness}
    Let $(\mathbb{A}, \mathbb{B})$ be a $\PCSP$ \kl{template} such that $\, \Pol(\mathbb{A}, \mathbb{B})$ satisfies the \kl{random 2-to-1 condition}. Then $\PCSP(\mathbb{A}, \mathbb{B})$ (the decision version) is $\NP$-hard, assuming the \kl{Rich 2-to-1 Conjecture}.
\end{theorem}

No matter how weak the hardness conditions are established, clearly there must be a point where we cannot hope for an appropriate choice function anymore. The extreme scenario is when $\Pol(\A, \B)$ contains \intro{symmetric} functions of arbitrarily large arities, that is, invariant over permutations of variables; since all coordinates are ``identical'', there is no indication for a choice function. Luckily, in this case we know that the $\PCSP$ (in the decision version) is solvable in polynomial time via an algorithm based on linear programming relaxation \cite{BLP_tractability}. Interestingly, symmetric polymorphisms do not take part in the algorithm itself, but ensure correctness of the rounding phase. To recover the solution and solve the search version of $\PCSP$, one needs a bit more: Brakensiek and Guruswami \cite{LP} provided an algorithm for search $\PCSP$s with polymorphism minions containing some specific, \textit{structured} subsets of functions. An example of such a set of structured functions is the set of all \textit{threshold functions} with a fixed threshold boundary.

\begin{definition}[\intro{Threshold function}]
    We say that $f : \Bool^n \to \Bool$ is a \textit{threshold function} if there is a number $t$ such that 
    \[
        f(\tuple x) = \begin{cases}
            1 & \text{ if } \tuple x \text{ has at least $t$ ones}, \\
            0 & \text{ otherwise.}
        \end{cases}
    \]
\end{definition}

For any $q \in (0,1)$, by $\THR_q$ we denote the set of all threshold functions that admit value $1$ if and only if at least a $q$-fraction of input values is equal to $1$. By \cite{LP} we know that if $\THR_q \subseteq \Pol(\A, \B)$ for any $q$, then the search version of $\Pol(\A, \B)$ is solvable in polynomial time. However, if $\Pol(\A, \B)$ contains symmetric functions of arbitrarily large arities, but we have no more guarantees about their structure, then the algorithm of \cite{BLP_tractability} is only applicable to the decision version of $\PCSP(\A, \B)$.

\begin{theorem}[Theorem 2 in \cite{BLP_tractability}]\label{theorem:blp_for_pcsps_with_symmetric_polymorphisms}
    Suppose that $(\mathbb{A}, \mathbb{B})$ is a $\PCSP$ template such that $\Pol(\mathbb{A}, \mathbb{B})$ has symmetric polymorphisms of arbitrarily large arities. Then the BLP+Affine algorithm solves $\PCSP(\mathbb{A}, \mathbb{B})$ (the decision version) in polynomial time.
\end{theorem}

We will use \cref{theorem:blp_for_pcsps_with_symmetric_polymorphisms} for the tractability part of our dichotomy theorem in \cref{sec:positive_polynomials}. Before we proceed to stating our results, we need to introduce one last piece of theory --- the \textit{Fourier Analysis of Boolean functions}, which provides crucial tools for analysis of Boolean polymorphisms.  

\subsection{Fourier Analysis and Polynomial Threshold Functions}

\intro{Fourier analysis of Boolean functions} is a field of research which developed a plethora of powerful tools to study the nature and behavior of various classes of Boolean functions, i.e. functions of the form $f: \{0,1\}^n \to \{0,1\}$. The heart of Fourier analysis is \kl{Fourier decomposition}. Briefly speaking, it is a representation of a given function as an element of a vector space spanned by an orthonormal basis consisting of functions indexed by subsets of coordinates. We defer the formalization of these notions to \cref{sec:influence}, but the important thing is that we can deduce a lot about the nature of the function from the coefficients of this decomposition. We point to \cite{kkl, friedguttheorem, Mossel2005NoiseSO, hypergraphremovallemmas, global_hypercontractivity} for some examples of prominent results in the field of \kl{Fourier analysis}.

In this paper, we attempt to apply the theory of Fourier analysis to polymorphisms of Boolean $\PCSP$s. As we have seen, the bulk of hardness results for $\PCSP$s consists of finding an appropriate choice function that selects some distinguished coordinates in every polymorphism. The difficulty in finding such a function lies in the fact that the selected coordinates must be compatible with the operation of taking minors. This motivates the search for measures of coordinate significance, which would be transported through minor maps. A natural fit for such a measure seems to be one of the central notions of Fourier analysis: the \textit{influence}. If $p \in (0,1)$, by $\cube{p}$ we denote the (product) $p$-\textit{biased distribution} over $\{0,1\}^n$, where each bit is independently set to $1$ with probability $p$ and to $0$ with probability $1-p$. For a tuple $\tuple x \in \{0,1\}^n$ and $s \in \{0,1\}$, by $\tuple x^{i \to s}$ we denote the tuple equal to $\tuple x$ with the value of $x_i$ set to $s$.

\begin{restatable}[\intro{Influence}]{definition}{definitionofinfluence}
    Suppose that $f: \Bool^n \to \mathbb{R}$ and $i \in [n]$. The \kl{influence} of coordinate $i$ in $f$ over the $p$-biased distribution is defined as 
    \[
        \Inf[(p)]{f, i} = \EX_{\tuple x \sim \cube{p}} \bigg[\Big(f(\tuple x) - \EX_{s \sim \cube{p}} \big[  f(\tuple x^{\, i \to s}) \big] \Big)^2 \bigg].
    \]
    The \intro{total influence} of $f$ is the sum of individual influences of all coordinates, i.e. $\I[(p)]{f} = \sum_{i=1}^n \Inf[(p)]{f,i}$.     
\end{restatable}

Simply put, the coordinate influence measures the probability that changing the input value on this coordinate changes the function output. Sadly, it is not difficult to construct an example of a \kl{minor map} that does not preserve significant coordinate influence in the original function.

\begin{example}[Influence collapse]\label{example:minors_destroy_significance}
    Let $f(x_1, x_2, x_3, x_4)$ be a function defined as follows:
    $$
    f(x_1, x_2, x_3, x_4) = \begin{cases}
        0 & \text{if $x_2 = x_3 = x_4$}, \\
        x_1 & \text{otherwise}
    \end{cases}
    $$
    We can observe that $\Inf[(1/2)]{f, 1} = 3/4$, while all other coordinates have influence $1/4$. A naive way of defining a choice function $C$ based on this fact would be to set $C(f) = \{1\}$, since $1$ has a significantly larger influence than all other coordinates. However, identifying coordinates $2,3,4$ yields a minor in which $1$ has influence equal to $0$. 
\end{example}

However, we are able to show that in a certain class of Boolean functions, significant coordinate influence is preserved with constant probability over a uniformly random \kl{2-to-1} minor. This result is strongly inspired by \cite[Lemma 23]{braverman_et_al:LIPIcs.ITCS.2021.27}, although it has a slightly different flavor. An analogous result has been obtained for the measure of so-called \textit{Shapley values} in \cite{brakensiek_et_al:LIPIcs.ICALP.2021.37}, but its drawback is that it only applies to \kl{monotone} Boolean functions. Our result allows us to discard monotonicity in exchange for a different assumption: that the sum of coordinate influences is not too large (see \cref{proposition:influence:our_main_result}). We are able to apply our tools to $\PCSP$s whose set of polymorphisms consists of functions from a certain class of restricted complexity, called \textit{Polynomial Threshold Functions} ($\PTF$s), which is a broader class of functions than \kl{threshold functions}.

\begin{definition}[\intro{Polynomial Threshold Function}] 
    We say that $f : \{0,1\}^n \to \{0,1\}$ is a \textit{Polynomial Threshold Function} ($\PTF$) if there exists a multilinear\footnote{I.e. a polynomial whose all terms consist of variables with degree at most $1$.} polynomial $Q : \{0,1\}^n \to \mathbb{R}$ with real coefficients, such that 
    \[
        f(\tuple x) = \begin{cases}
            1 & \text{ if } Q(\tuple x) \geq 0, \\ 
            0 & \text{ otherwise.}
        \end{cases}
    \]
    If $Q$ has degree $k$, then $f$ is a $\PTF$ of degree $k$. For every $k \geq 0$, by $\PTF_k$ we denote the class of all $\PTF$s of degree $k$.
\end{definition}

Observe that $\PTF_k$ is a \kl{minion} for every $k \geq 0$. This follows from the fact that if $Q$ is a polynomial, then all its minors are polynomials of degree not larger than the degree of $Q$. $\PTF$s have found significant interest in learning theory \cite{ptf_learning_1, learning_ptf_2, learning_ptf_3, learning_ptf_4, learning_ptf_5} and have been studied from a structural and extremal perspective \cite{extremal_properties_of_ptfs, Harsha2009BoundingTS, bounding_ptfs_2}. In the context of $\PCSP$s, a notable result is the complete dichotomy \cite[Theorem 1.2]{injective_hardness_condition} for $\PTF$s of degree $1$, also known as \intro{Linear Threshold Functions} ($\LTF$s), obtained with the \kl{injective layered condition}.

The property of $\PTF$s that interests us the most is that they admit strong stability to perturbations of input values, which has been independently proven in \cite[Theorem 1.3]{Harsha2009BoundingTS} and \cite[Theorem 1.3]{bounding_ptfs_2}. As we shall see in \cref{sec:influence}, this implies the so-called \textit{low-degree concentration} of $\PTF$s, which makes them ideal candidates for the application of our results on influence translation. 

\subsection{Our contribution}

We believe that $\PTF$s of bounded degree are good candidates for Boolean functions that could be fully understood in the context of $\PCSP$s, and our work takes the next step towards this goal. Although we were unable to achieve a complete characterization as authors of \cite{injective_hardness_condition} managed to do for the special case of $\LTF$s, an additional assumption of polynomial threshold representations having non-negative coefficients proved to be enough for a complete characterization. We present our dichotomy result in an informal and shortened version; see \cref{sec:positive_polynomials} for the full statement.

\begin{theorem}\label{theorem:dichotomy_for_positive_polynomials}
    Suppose that $(\A, \B)$ is a $\PCSP$ \kl{template} such that $\Pol(\A, \B)$ consists of \kl{Polynomial Threshold Functions} admitting representations with non-negative coefficients of bounded degree. Then $\PCSP(\A, \B)$ (the decision version) is in P or is NP-complete.
\end{theorem}

Since $\PTF$s admitting representations with non-negative coefficients are necessarily monotone, \cref{theorem:dichotomy_for_positive_polynomials} can be seen as a step toward a potential strengthening of the conditional dichotomy for Boolean $\PCSP$s with monotone polymorphisms of \cite{brakensiek_et_al:LIPIcs.ICALP.2021.37}.

\vspace{3mm}

\begin{figure}[hbtp]
    \centering
    \includegraphics[scale=0.7]{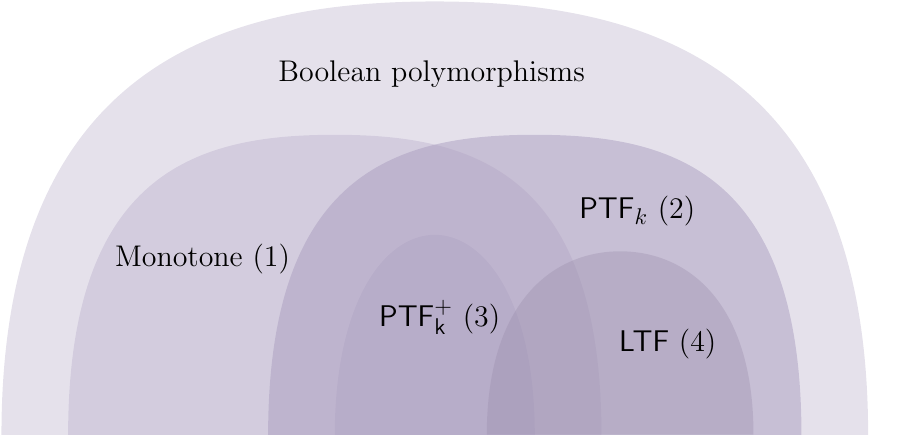}
    \caption{
        Currently there are two notable characterizations of Boolean $\PCSP$s with structurally restricted polymorphisms: (1) conditional dichotomy for monotone polymorphisms \cite{brakensiek_et_al:LIPIcs.ICALP.2021.37} and (4) dichotomy for $\LTF$s \cite{injective_hardness_condition}. We contribute to this landscape with (3) a dichotomy for $\PTF$s of bounded degree with non-negative coefficients (denoted by $\PTF_k^+$ on the figure) with \cref{theorem:dichotomy_for_positive_polynomials}, and (2) a hardness result for $\PTF$s of bounded degree with influential coordinates with \cref{theorem:hardness_for_ptfs_with_large_influences}. 
    }
\end{figure}

We also provide some insight into minions with $\PTF$s of bounded degree with possibly negative coefficients. Under the assumption of the \kl{Rich 2-to-1 Conjecture}, we obtain the following hardness result.

\begin{theorem}\label{theorem:hardness_for_ptfs_with_large_influences}
    Suppose that $(\A, \B)$ is a $\PCSP$ \kl{template} such that $\Pol(\A, \B) \subseteq \PTF_k$ for some $k \geq 0$. If
    \[
        \exists \, \delta > 0 : \forall \, f \in \minion : \exists \, i : \Inf[(1/2)]{f, i} \geq \delta,
    \]
    then $\PCSP(\A, \B)$ (the decision version) is $\NP$-hard, assuming the \kl{Rich 2-to-1 Conjecture}.
\end{theorem}

\cref{theorem:hardness_for_ptfs_with_large_influences} can also be derived from \cite[Lemma 23]{braverman_et_al:LIPIcs.ITCS.2021.27} with a hardness reduction similar to ours. However, we believe that the analytical tools we develop in our proof might be of independent interest for future research; in particular, the following result on \kl{influence} translation in $p$-\kl{biased distributions}, which is our main technical lemma.

\begin{restatable}{proposition}{influencepreservation}\label{proposition:influence:our_main_result}
    Suppose that $\lambda \in (0,1/2)$ and $\delta > 0$. There are positive constants $\gamma = \gamma(\lambda, \delta)$, $\tau =~\tau(\lambda, \delta)$ and $n_0 = n_0(\lambda, \delta)$ such that the following holds. Suppose that $f : \{0,1\}^{2n} \to \{0,1\}$ is a Boolean function with $n \geq n_0$ and $p \in (\lambda, 1-\lambda)$. If $\, \I[(p)]{f} \leq \gamma \cdot n$, then 
    \[
        \forall \, i \in [2n] : \Inf[(p)]{f, i} \geq \delta \implies \Pr \Big[ \Inf[(p)]{f^\pi, \pi(i)} \geq \tau \Big] \geq \tau,
    \]
    where $\pi : [2n] \to [n]$ is a uniformly random $2$-to-$1$ minor map.
\end{restatable}

\subsection{Organization of the paper}

We start with \cref{sec:random_condition}, in which we illustrate the standard reduction approach from \kl{Label Cover} problems to $\PCSP$s and prove \cref{theorem:reductions:random_condition_implies_hardness} by providing an explicit construction for $\PCSP$s satisfying the \kl{random 2-to-1 condition}.

The following two sections are focused on our contributions. \cref{sec:positive_polynomials} introduces \textit{Positive Polynomial Threshold Functions} and establishes relevant notation. We then proceed to prove \cref{theorem:dichotomy_for_positive_polynomials}. Moving forward, in \cref{sec:influence}, we outline essential concepts of the $p$-biased \kl{Fourier analysis} of Boolean functions. Building on these, we develop the necessary technical tools and establish \cref{theorem:hardness_for_ptfs_with_large_influences}.

Finally, \cref{sec:future} explores potential extensions of our results and directions for future research on Boolean $\PCSP$s.

\section{Hardness from random 2-to-1 condition}\label{sec:random_condition}
In this section, we give an overview of how hardness conditions for $\PCSP$s are typically obtained. Our goal is to show that \kl{Rich 2-to-1 Gap Label Cover} can be reduced to any $\PCSP$ satisfying \kl{random 2-to-1 condition}, which will immediately imply \cref{theorem:reductions:random_condition_implies_hardness}. We note that this result has already been shown in \cite[Section 4.2]{brakensiek_et_al:LIPIcs.ICALP.2021.37} in the special case of Boolean $\PCSP$s with monotone polymorphisms --- we describe it in a more general setting. The reduction follows a standard framework and is split into two parts; this approach was originally established in \cite{algebraic_approach_to_pcsps}.

Before we proceed to the construction, we introduce the necessary notions. A \intro{minor condition} is a finite set of \intro{identities} of the form
\[
    f(x_{\pi(1)}, \dots, x_{\pi(n)}) \approx g(x_1, \dots, x_m),
\]
where $f$ and $g$ are functional symbols and $\pi : [n] \to [m]$. In addition, we require that the sets of symbols on left and right-hand sides are disjoint. Given a minion $\minion$ on $(A, B)$, we say that an identity is \intro{satisfied} in $\minion$ if there exists a \intro{symbol interpretation} $\zeta$ that assigns a function in $\minion$ to every symbol, such that 
\[
    \forall x_1, \dots, x_n \in A : \zeta(f)(x_{\pi(1)}, \dots, x_{\pi(n)}) = \zeta(g)(x_1, \dots, x_m).
\]
Furthermore, a minor condition is satisfied in $\minion$ if there exists an interpretation that satisfies all identities simultaneously.

A minor condition is \intro{trivial} if it is satisfied in every \kl{minion}. We note that if a minor condition $\Sigma$ is satisfied in a minion of \kl{projections} on a set of at least two elements, then it is automatically trivial. To see why, observe that if such a satisfying interpretation assigns projections $\zeta(f)(\tuple x) = x_i$ and $\zeta(g)(\tuple y) = y_j$, then it must hold that $x_i = x_{\pi(j)}$ for every choice of $\tuple x$, and since the underlying set has a size of at least two, this implies $i = \pi(j)$. Now, to find a satisfactory interpretation of $\Sigma$ in another minion on $(A, B)$, we use a small trick: take a unary function $h : A \to B$ in the minion and generate all minors of $h$ of the form $(x_1, \dots, x_n) \mapsto h(x_i)$. The functions obtained are not exactly projections, but behave similarly. In particular, the interpretation $\zeta'$ defined as $\zeta'(f)(\tuple x) = h(x_i)$ if and only if $\zeta(f)(\tuple x) = x_i$, satisfies $\Sigma$. 

We now define the intermediate problem in our reduction, called \textit{Promise Minor Condition}.

\begin{definition}[\intro{Promise Minor Condition}]
    Given $n \geq 1$ and a minion $\minion$, the \textit{Promise Minor Condition} problem $\PMC_n(\minion)$ is the problem defined as follows: for an input \kl{minor condition} $\Sigma$ with \kl{identities} of arity at most $n$, answer $\mathsf{YES}$ if $\Sigma$ is \kl{trivial}, and answer $\mathsf{NO}$ if $\Sigma$ is not \kl{satisfied} in $\minion$.
\end{definition}

We emphasize that the parameter $n$ is treated as a constant and the size of an instance is the number of identities. As it turns out, if $(\A, \B)$ is a $\PCSP$ \kl{template} and $\minion = \Pol(\A, \B)$, then $\PMC_n(\minion)$ is log-space reducible to $\PCSP(\A, \B)$. This fact follows from a construction that resembles \textit{long code tests}, and is one of the most fundamental results in algebraic approach to $\PCSP$s. We refer to \cite[Section 3.3]{algebraic_approach_to_pcsps} for more details.

The remaining part is to show how \kl{Rich 2-to-1 Gap Label Cover} can be reduced to \kl{Promise Minor Condition}. Suppose that $\Psi = (L \cup R, E, 2n, n, \Pi)$ is an instance of \kl{Rich 2-to-1 Label Cover}. We construct a minor condition from $\Psi$ in the following way. We identify every vertex $u \in L$ with a symbol $f_u$ of arity $2n$ and $v \in R$ with a symbol $g_v$ of arity $n$. Moreover, for every edge $e = (u, v)$ and its corresponding constraint $\pi_e \in \Pi$, we write the following identity:
\[
    f_u\big(x_{\pi_e(1)}, \dots, x_{\pi_e(2n)}\big) \approx g_v (x_1, \dots, x_n). \tag{{\color{magenta}$*$}} \label{eq:minor_condition_construction_from_label_cover}
\]
Denote the resulting minor condition by $\Sigma$. Observe that $\Sigma$ is essentially an equivalent presentation of $\Psi$: the \kl{labelings} of $\Psi$ are in 1-to-1 correspondence with the interpretations of $\Sigma$ in \kl{projections}. This automatically yields the \intro{completeness} of our reduction, i.e. preservation of $\mathsf{YES}$-instances --- if all \kl{constraints} of $\Psi$ can be simultaneously satisfied, then $\Sigma$ is \kl{trivial}.

It is the other side that is challenging: the preservation of $\mathsf{NO}$-instances, which is called \intro{soundness}. For this part to work, we need additional assumptions about $\minion$. In our case, this assumption is the \kl{random 2-to-1 condition}. We show how to utilize this condition in the proof of the following proposition, which, together with the fact that $\PMC_n(\Pol(\A, \B))$ reduces to $\PCSP(\A, \B)$, finishes the reduction and implies \cref{theorem:reductions:random_condition_implies_hardness}.

\begin{proposition}
    Let $\minion$ be a minion. If $\minion$ satisfies the \kl{random 2-to-1 condition}, then there is a positive constant $\varepsilon = \varepsilon(\minion)$ such that for every $n \geq 1$, $\GapRich_n[1, \varepsilon]$ poly-time reduces to $\PMC_{2n}(\minion)$.
\end{proposition}

\begin{proof}
    Fix a minion $\minion$, which satisfies the \kl{random 2-to-1 condition} with constants $M, \tau > 0$ and a choice function $C$. Let $\varepsilon = \tau/(2M^2)$ and fix any $n \geq 1$. We want to show that $\GapRich_n[1, \varepsilon]$ reduces to $\PMC_{2n}(\minion)$. We start with a \kl{Rich 2-to-1 Label Cover} instance $\Psi = (L \cup R, E, 2n, n, \Pi)$. Next, we construct a minor condition $\Sigma$ as in \eqref{eq:minor_condition_construction_from_label_cover}. It is obvious that $\Sigma$ can be constructed in polynomial time. Therefore, it remains to show that this is a proper reduction.

    \vspace{3mm}
    \noindent
    \textbf{Completeness.} Suppose that $\Psi$ is a $\mathsf{YES}$-instance of $\GapRich_n[1,\varepsilon]$ and all constraints in $\Pi$ are satisfied with a \kl{labeling} $\sigma : L \cup R \to [2n]$. As discussed above, $\sigma$ induces an \kl{interpretation} that assigns the \kl{projection} $\tuple x \mapsto x_{\sigma (w)}$ to the symbol corresponding to vertex $w$. This interpretation satisfies all identities, and thus $\Sigma$ is \kl{trivial}.

    \vspace{3mm}
    \noindent
    \textbf{Soundness.} By contraposition, suppose that $\Sigma$ is satisfied in $\minion$, which is witnessed by an interpretation $\zeta$. Our goal is to show there is a labeling $\sigma : L \cup R \to [2n]$, which satisfies more than $\varepsilon$-fraction of constraints in $\Pi$. Consider a random labeling $\sigma$, such that for every $u \in L$, the value $\sigma(u)$ is chosen uniformly at random from $C(\zeta(f_u))$. Similarly, $\sigma(v)$ is chosen uniformly at random from $C(\zeta(g_v))$ for every $v \in R$. 
    
    Observe that for every identity as in \eqref{eq:minor_condition_construction_from_label_cover}, the function $\zeta(g_v)$ is a 2-to-1 minor of $\zeta(f_u)$ with respect to the map $\pi_e$. Fix any $u \in L$ and let $\Sigma(u)$ be the set of identities in $\Sigma$ involving the symbol $f_u$. \kl{Richness} of $\Psi$ gives that if we choose, uniformly at random, an identity in $\Sigma(u)$, take the symbol $g_v$ on the right-hand side and return $\zeta(g_v)$, we obtain a uniform distribution over all 2-to-1 minors of $\zeta(f_u)$. From the \kl{random 2-to-1 condition}, we have that at least a $\tau$-fraction of identities in $\Sigma(u)$ satisfies $\pi_e(C(\zeta(f_u))) \cap C(\zeta(g_v)) \neq \emptyset$. For every such identity, since both sets $C(\zeta(f_u))$ and $C(\zeta(g_v))$ have sizes at most $M$, we have $\pi_e(\sigma(u)) = \sigma(v)$ with probability at least $1/M^2$ over the choice of $\sigma$. Hence, the expected fraction of satisfied constraints in $\Pi$ adjacent to $u$ is at least $\tau/M^2$. By summing over $u \in L$, we obtain that the expected number of constraints satisfied by $\sigma$ is at least $\tau/M^2 > \varepsilon$.
\end{proof}

\section{Dichotomy for Positive Polynomial Thresholds of bounded degree}\label{sec:positive_polynomials}
Our main goal in this section is to prove \cref{theorem:dichotomy_for_positive_polynomials}, but first we introduce \textit{Positive Polynomial Threshold Functions} and relevant notation. If $\tuple x \in \Bool^n$ and $i \in [n]$, then by $\tuple x \oplus i$ we denote the tuple $\tuple x$ with entry on $i$-th position flipped. Suppose that $Q : \{0,1\}^n \to \mathbb{R}$ is a multilinear polynomial. For every $S \subseteq [n]$, we let $\widehat{Q}(S)$ be the coefficient of $Q$ corresponding to term $\prod_{i \in S} x_i$. We say that $Q$ is \intro{unbiased} if the constant term of $Q$ is equal to $0$ and that it is \intro{positive} if all non-constant coefficients of $Q$ are non-negative.

\begin{definition}[\intro{Positive Polynomial Threshold Function}]
    A function $f : \{0,1\}^n \to \{0,1\}$ is a \textit{Positive Polynomial Threshold Function} ($\PTF^+)$ if there exists an \kl{unbiased}, and \kl{positive} polynomial $Q : \{0,1\} \to \mathbb{R}$ and $t \geq 0$ such that $f = \repr{Q}{t}$, where 
    \[
        \repr{Q}{t}(\tuple x) = \begin{cases}
            0 & \text{ if } Q(\tuple x) < t, \\ 
            1 & \text{ if } Q(\tuple x) \geq t.
        \end{cases}
    \]
    If $Q$ has degree $k$, we say that $\repr{Q}{t}$ is a \intro{representation} of $f$ of degree $k$ and that $f$ is a $\PTF^+$ of degree $k$. For every $k \geq 0$, by $\PTF_k^+$ we denote the class of \kl{Positive Polynomial Thresholds} of degree $k$.
\end{definition}

For the remainder of this section, we will assume that all representations are \intro{normalized}, i.e. that the sum of coefficients of the underlying polynomial is $1$. Observe that every $\PTF^+$ must be \kl{monotone}. Furthermore, it is clear that every $\PTF^+$ is a $\PTF$ of the same degree. In other words, $\PTF_k^+ \subseteq \PTF_k$ for every $k \geq 0$.

We now introduce the notion of coordinate significance that we utilize in our dichotomy result --- the \textit{weight} which is simply the sum of coefficients corresponding to terms including the coordinate.

\begin{definition}[\intro{Weight of coordinate}]
    Suppose that $Q : \{0,1\}^n \to \mathbb{R}$ is a multilinear polynomial. We define the \textit{weight} of coordinate $i$ in $Q$ as 
    \[
        \Deg{Q}{i} = \sum_{S \ni i} \widehat{Q}(S). 
    \]
\end{definition}

We are in a position to reveal the boundary between tractable and $\NP$-hard problems in \cref{theorem:dichotomy_for_positive_polynomials}. The condition that captures the dichotomy boundary is the following \textit{regularity condition}.

\begin{condition}[\intro{Regularity condition}]
    Suppose that $k \geq 1$ and $\minion \subseteq \PTF_k^+$ is a minion. We say that $\minion$ satisfies the $k$-regularity condition if
    \[
        \forall \, \varepsilon > 0 : \exists \, t \geq 0, Q: \{0,1\}^n \to \mathbb{R} \text{ of degree $k$} : \Big( \repr{Q}{t} \in \minion \, \text{ and } \, \max_{i \in [n]} \, \Deg{Q}{i} < \varepsilon \Big).
    \]
\end{condition}

\setcounter{section}{1}
\setcounter{theorem}{19}
\begin{theorem}[Extended version]\label{theorem:dichotomy_for_positive_ptfs_extended}
    Suppose that $k \geq 0$ and $(\mathbb{A}, \mathbb{B})$ is a Boolean $\PCSP$ template such that $\Pol(\mathbb{A}, \mathbb{B}) \subseteq \PTF_k^+$. Then
    \begin{enumerate}[label=(\arabic*)] 
        \item either $\Pol(\mathbb{A}, \mathbb{B})$ satisfies the $k$-\kl{regularity condition} and $\PCSP(\mathbb{A}, \mathbb{B})$ is tractable, or
        \item $\Pol(\mathbb{A}, \mathbb{B})$ does not satisy the $k$-\kl{regularity condition} and $\PCSP(\mathbb{A}, \mathbb{B})$ is $\NP$-complete,
    \end{enumerate}
    where in both cases $\PCSP(\A, \B)$ is the decision version of the problem.
\end{theorem}
\setcounter{section}{3}
\setcounter{theorem}{3}

We divide the proof of \cref{theorem:dichotomy_for_positive_ptfs_extended} into two parts. For tractability, we show that if a minion $\minion$ satisfies the \kl{regularity condition}, then it contains \kl{threshold functions} of arbitrarily large arities --- we achieve this with the help of results of \cite{brakensiek_et_al:LIPIcs.ICALP.2021.37}, which establish a connection between threshold functions and so-called \textit{Shapley values}. This is enough to imply tractability due to \cref{theorem:blp_for_pcsps_with_symmetric_polymorphisms}. 

On the other hand, we can show that if $\minion$ does not satisfy the \kl{regularity condition}, then it must satisfy the \kl{layered choice condition}, which implies hardness by \cref{theorem:reductions:layered_condition_implies_hardness}. These two results are encapsulated in \cref{proposition:positive_polynomials:tractability} and \cref{proposition:positive_polynomials:hardness}, which we prove in the following two subsections.

\subsection{Tractability from regularity condition}

The main focus of this subsection is to prove the following result, which implies the tractability part of \cref{theorem:dichotomy_for_positive_polynomials} by \cref{theorem:blp_for_pcsps_with_symmetric_polymorphisms}.

\begin{proposition}\label{proposition:positive_polynomials:tractability}
    Suppose $k \geq 1$ is an integer and $\minion \subseteq \PTF_k^+$ is a minion satisfying the $k$-\kl{regularity condition}. Then $\minion$ contains threshold functions of arbitrarily large degree.
\end{proposition}

Intuitively, the \kl{regularity condition} asserts that there are arbitrarily ``regular'' functions in $\minion$, i.e. such that no coordinate is visibly more significant than others. This intuition of regularity in general functions can be captured by various notions of coordinate significance, which are not always equivalent, or even comparable. In this section, we focus on one of such notions, called \textit{Shapley values}. We provide the following combinatorial definition of Shapley values from \cite[Definition 2.7]{brakensiek_et_al:LIPIcs.ICALP.2021.37}, but note that there are also equivalent analytical definitions \cite[Exercise 8.31]{O’Donnell_2014}. 

\begin{definition}[\intro{Shapley values}]
  Suppose $f : \{0,1\}^n \to \{0,1\}$ is a \kl{monotone} Boolean function and $i \in [n]$. Consider the following probabilistic experiment. Draw uniformly a permutation $\sigma$ of elements of $[n]$. Let $\tuple x_0 = (0, \dots, 0)$ be of length $n$ and $\tuple x_j = \tuple x_{j-1} \oplus \sigma(j)$ for every $j \in [n]$. The \kl{Shapley value} of coordinate $i$ in $f$ is defined as
    \[
      \Phi_i[f] = \Pr\Big[f\big(\tuple x_{\sigma^{-1}(i)-1}\big) = 0 \text{ and } f\big(\tuple x_{\sigma^{-1}(i)}\big) = 1\Big].
    \]
\end{definition}

In other words, the \kl{Shapley value} of coordinate $i$ is the probability that when flipping bits from $0$ to $1$ in order according to a uniformly drawn permutation of $[n]$, flipping the value on $i$-th coordinate is a moment when value of $f$ changes from $0$ to $1$. Observe that if $f$ is not constant, then there is exactly one such moment for every permutation --- and therefore the sum of \kl{Shapley values} of coordinates in $f$ must be $1$. 

The importance of the notion of Shapley values lies in its connection with the so-called \textit{sharp threshold phenomenon}. Suppose that $f: \{0,1\}^n \to \{0,1\}$ is a Boolean function and $p \in [0,1]$. We consider the function $\zeta : p \mapsto \EX_p[f]$, where $\EX_p[f]$ is the expected value of $f$ in the product $p$-biased distribution. It is known that if $f$ is non-constant and monotone, then $\zeta$ is a continuous, strictly increasing function. We say that $f$ admits a \textit{sharp threshold} if $\zeta(p)$ grows from values close to $0$ to values close to $1$ in a short interval around the \textit{critical probability}, i.e. the value $q \in (0,1)$ for which $\EX_{q}[f] = 1/2$. For every $\varepsilon > 0$, we say that $[s, t]$ is the $\varepsilon$-threshold interval if $\EX_s[f] = \varepsilon$ and $\EX_t[f] = 1-\varepsilon$. As long as the critical probability is bounded away from $0$ and $1$ by a constant, it is well understood what classes of monotone functions admit arbitrarily small threshold intervals for all $\varepsilon > 0$. It turns out that they are exactly the classes of functions with arbitrarily small \kl{Shapley values} \cite[Theorems 3.3 and 3.5]{kalaiSocialIndeterminancy}. 

The phenomenon of sharp thresholds has also found its usage in the context of $\PCSP$s. The simplest class of Boolean monotone functions exhibiting arbitrarily sharp thresholds are the \textit{majority functions}, or more generally, \kl{threshold functions} (this fact follows from classic concentration inequalities for binomial distribution). Therefore, it seems plausible to conjecture that as long as a Boolean minion $\minion$ contains functions of arbitrarily small threshold intervals (or equivalently, arbitrarily small \kl{Shapley values}), then $\minion$ must contain threshold functions of arbitrary large arities. This turns out to be true, as has been proven with a probabilistic argument in \cite{brakensiek_et_al:LIPIcs.ICALP.2021.37}.

\begin{lemma}[Lemma 3.4 in \cite{brakensiek_et_al:LIPIcs.ICALP.2021.37}]\label{lemma:positive_polynomials_small_sv_give_thresholds}
    For every $m \geq 2$, there exists a constant $\delta = \delta(m) > 0$ such that the following holds. Suppose $f : \{0,1\}^n \to \{0,1\}$ is a \kl{monotone} Boolean function such that
    \[
        \max_{i \in [n]} \, \Phi_i[f] \leq \delta.
    \]
    Then $f$ has a \kl{threshold function} of arity at least $m$ as a \kl{minor}.
\end{lemma}

To prove \cref{proposition:positive_polynomials:tractability}, we will show that \kl{regularity condition} implies that the minion contains functions with arbitrarily small \kl{Shapley values}. To this end, we will use probabilitic method to find minors in $\minion$ that satisfy even stronger regularity assumptions than \kl{regularity condition}: that every two coefficients of their representations corresponding to sets of the same size are almost equal. This turns out to be a sufficient condition for small \kl{Shapley values}.

\begin{lemma}\label{lemma:positive_polynomials:similar_edges_give_small_sv}
    For every $k \geq 1$ and $\delta > 0$, there exist constants $m = m(k, \delta) \geq 1$ and $\tau = \tau(k, \delta) > 0$, such that the following holds. Suppose that $f : \{0,1\}^m \to \{0,1\}$ is a \kl{monotone} Boolean function with \kl{representation} $f = \repr{Q}{t}$ of degree at most $k$. Then 
    \[
        \bigg( \forall \, S,T \subseteq [m] \text{ with } |S| = |T| : \Big| \widehat{Q}(S)-  \widehat{Q}(T) \Big| \leq \tau \bigg) \implies \max_{i \in [m]} \,\Phi_i[f] \leq \delta.
    \]
\end{lemma}

\begin{proof}
  Fix values of $k \geq 1$ and $\delta > 0$. We define constants $m, \tau$ such that they satisfy the following properties:
  \begin{equation*}
    m^{-\frac{1}{2k}} \leq \delta/3, \text{\hspace*{1cm}} m^{-\frac{1}{2}} > \frac{2k^2}{m}, \text{\hspace*{1cm}} m^{\frac{2k-1}{2k}} \geq k, \text{\hspace*{1cm}} \tau = \frac{k}{m \cdot 2^{m+1}}. \tag{{\color{magenta}$*$}}\label{eq:positive_polynomials:assumptions}
  \end{equation*}
  Observe that the first three conditions of \eqref{eq:positive_polynomials:assumptions} can be satisfied with sufficiently large $m = m(k, \delta)$. Suppose that $f: \{0,1\}^m \to \{0,1\}$ is a monotone Boolean function as in the lemma statement. We will show that \kl{Shapley values} of $f$ are at most $\delta$.

  We start with the observation that we can bound \kl{weights} of coordinates. First, we bound the difference of weights of any two coordinates $i, j \in [m]$:
  \begin{align*}
    \Big| \Deg{Q}{i} - \Deg{Q}{j} \Big| &= \bigg| \sum \Big\{    \widehat{Q}(S) : i \in S \text{ and } j \not \in S \Big\} - \sum \Big\{ \widehat{Q}(T) : i \not \in S  \text{ and } j \in S \Big\} \bigg| \\  
                                        &= \bigg| \sum \Big\{ \widehat{Q}(S) - \widehat{Q}\big((S \setminus \{i\}) \cup \{j\} \big) : i \in S \text{ and } j \not \in S \Big\} \bigg| \leq 2^m \cdot \tau.
  \end{align*}
  Since $Q$ is \kl{normalized} and its degree is at most $k$, the sum of weights of coordinates is at most $k$. Therefore, the average weight is at most $k/m$ and so is the minimal weight. From the bound for differences of weights we obtain that
  \[
    \max_{i \in [m]} \, \Deg{Q}{i} \leq (k/m) + 2^m \cdot \tau.
  \]
  Futhermore, since $Q$ is \kl{normalized}, there must be $\ell \in [k]$ such that the sum of weights of coefficients corresponding to sets of size $\ell$ is at least $1/k$. With a reasoning similar to the previous bound, for every $S \subseteq [m]$ of size $\ell$, we have
  \[
    \widehat{Q}(S) \geq \frac{1}{k \cdot {m \choose \ell}} - \tau.
  \]
  Fix any $i \in [m]$. We proceed to bounding the \kl{Shapley value} of $i$ in $f$. Consider the set $\mathcal{S}$ of permutations of $[m] \setminus \{i \}$. For every $\sigma \in \mathcal{S}$, consider the sequence of its values $\sigma(1), \sigma(2), \dots, \sigma(m-1)$. Let $s(\sigma) \in \{0, \dots, m \}$ be the number of positions where $i$ can be inserted into this sequence, so that the resulting permutation of $[m]$ counts towards $\Phi_i[f]$, i.e. flipping the bit on coordinate $i$ is the moment when value of $f$ changes from $0$ to $1$, if bits are flipped in order corresponding to this permutation. It turns out that $s(\sigma)$ cannot be large; we will show that it is $o(m)$.
  
  Fix any $\sigma \in \mathcal{S}$. We can assume that $s(\sigma) > 0$. Observe that the set of positions counted by $s(\sigma)$ is a range: denote it by $[a, b]$, i.e. for every $j \in [a,b]$, if we insert $i$ after the element $\sigma(j)$, we obtain a permutation counting towards $\Phi_i[f]$ (if $a = 0$, then $i$ can be inserted before $\sigma(1)$). Let $\tuple x \in \{0,1\}^m$ be the characteristic vector of $\{ \sigma(1), \sigma(2), \dots \sigma(a) \}$, i.e.
  \[
    \tuple x_j = \begin{cases}
        1 & \text{ if } j \in \{ \sigma(1), \sigma(2), \dots, \sigma(a) \}, \\ 
        0 & \text{ otherwise}.
    \end{cases} 
  \]
  Similarly, let $\tuple y$ be the characteristic vector of $\{ \sigma(1), \sigma(2), \dots, \sigma(b) \}$. We have $f(\tuple x) = f(\tuple y) = 0$ and $f(\tuple x \oplus i) = f(\tuple y \oplus i) = 1$. In particular, this means that $Q(\tuple x), Q(\tuple y) < t$ but $Q(\tuple x \oplus i), Q(\tuple y \oplus i) \geq t$. We obtain that
  \[
     Q(\tuple y) - Q(\tuple x) < Q(\tuple x \oplus i) - Q(\tuple x)
  \]
  Observe that the left-hand side of this inequality is at least the sum of coefficients of $Q$ corresponding to subsets of $\{ \sigma(a+1), \dots, \sigma(b) \}$, while the right-hand side is at most $\Deg{Q}{i}$. Let $c = b-a$. There are ${c \choose \ell}$ subsets of size $\ell$ of $\{\sigma(a+1), \dots, \sigma(b) \}$ and we known that all of their coefficients in $Q$ are large. We obtain the following inequality:
  \[
      {c \choose \ell} \cdot \left( \frac{1}{k \cdot {m \choose \ell}} - \tau \right) < \frac{k}{m} + 2^m \cdot \tau.
  \]
  We bound ${c \choose \ell} \leq 2^m$ and plug the value of $\tau$ from \eqref{eq:positive_polynomials:assumptions} to further obtain that 
  \begin{equation*}
    \frac{{c \choose \ell}}{{m \choose \ell}} < k \cdot \left( \frac{k}{m} + 2^{m+1} \cdot \tau \right) = \frac{2k^2}{m}. \tag{{\color{magenta}$**$}}\label{eq:positive_polynomials:bound_for_ration_of_binomial_coeffs}
  \end{equation*}
  Let $\alpha = \frac{2k-1}{2k} \in (0,1)$. Note that $m^\alpha \geq k$ from \eqref{eq:positive_polynomials:assumptions}. We argue that $c < 2m^\alpha$. Indeed, otherwise we would have that
  \begin{align*}
     \frac{{c \choose \ell}}{{m \choose \ell}} &= \frac{(c-\ell+1) \cdot (c-\ell+2) \dots (c-1) \cdot c}{(m-\ell+1) \cdot (m-\ell+2) \dots (m-1) \cdot m} \geq \left( \frac{c-k}{m} \right)^k \geq m^{(\alpha - 1) \cdot k} = m^{-\frac{1}{2}} \overset{\eqref{eq:positive_polynomials:assumptions}}{>} \frac{2k^2}{m},
  \end{align*}
  which is a contradiction with \eqref{eq:positive_polynomials:bound_for_ration_of_binomial_coeffs}. Since $s(\sigma) = c + 1$, we have $s(\sigma) < 2m^\alpha + 1 \leq 3m^\alpha$ for every $\sigma \in \mathcal{S}$. We are ready to bound $\Phi_i[f]$:
  \[
    \Phi_i[f] = \frac{1}{m!} \cdot \sum_{\sigma \in \mathcal{S}} s(\sigma) < \frac{1}{m!} \cdot (m-1)! \cdot 3m^\alpha = 3m^{\alpha - 1} = 3m^{- \frac{1}{2k}} \overset{\eqref{eq:positive_polynomials:assumptions}}{\leq} \delta.
  \]
  The coordinate $i$ was chosen arbitrarily from $[m]$, therefore we showed that every coordinate has \kl{Shapley value} at most $\delta$ in $f$, which finishes the proof.
\end{proof}

Before we proceed to the main proof, we introduce the last piece of the puzzle, the classic \textit{McDiarmid's concentration inequality}, otherwise known as \textit{bounded differences inequality}. This will come in handy in analysis of random minors of functions satisfying the \kl{regularity condition}.

\begin{theorem}[\intro{McDiarmid's Inequality}, Lemma (1.2) in \cite{McDiarmid_1989}]
    Suppose $m,n \geq 1$ are integers and $g~:~[m]^n\to \mathbb{R}$ satisfies the \intro{bounded differences property} with bounds $c_1, c_2, \dots, c_n \geq 0$, i.e. for every $i \in [n]$ and $\tuple x, \tuple x' \in [m]^n$ which differ only in the $i$-th coordinate, we have $|g(\tuple x) - g(\tuple x')| \leq c_i$.
    Consider independent random variables $X_1, X_2, \dots, X_n \in [m]$. Let $\mu = \EX[g(X_1, \dots, X_n)]$. Then for any $t > 0$,
    \[
        \Pr\Big[ \big\lvert g(X_1, \dots, X_n) - \mu \big\rvert \geq t \Big] \leq 2 \exp\left( -\frac{2t^2}{\sum_{i = 1}^n c_i^2} \right).
    \]
\end{theorem}

\begin{proof}[Proof of \cref{proposition:positive_polynomials:tractability}]
  Suppose $k \geq 1$ and $\minion \subseteq \PTF_k^+$ is a Boolean minion satisfying the $k$-\kl{regularity condition}. Our goal is to show that $\minion$ contains threshold functions of arbitrary large arities. Thanks to \cref{lemma:positive_polynomials_small_sv_give_thresholds}, it is sufficient to show that $\minion$ contains functions with arbitrary small \kl{Shapley values}.
  
  Fix any $\delta > 0$. We want to find a function in $\minion$ with all \kl{Shapley values} bounded by $\delta$. Let $m,\tau$ be constants from \cref{lemma:positive_polynomials:similar_edges_give_small_sv} applied with $k$ and $\delta$. Additionally, let
  \[
    \varepsilon = \frac{\tau^2}{2\ln 2 \cdot k(m+1)}
  \]
  and let $f \in \minion$ be a function with representation $f = \repr{Q}{t}$ of degree at most $k$, such that $\Deg{Q}{i} < \varepsilon$ for every $i \in [n]$. The existence of such function is asserted by $k$-\kl{regularity condition} for $\minion$.
  
  As the next step, we construct a random $m$-ary minor of $f$ as follows. Let $\pi : [n] \to [m]$ be a random minor map obtained by drawing $\pi(i)$ uniformly from $[m]$, idependently for every $i \in [n]$. Finally, let $g = f^\pi$. Observe that the \kl{representation} $f = \repr{Q}{t}$ induces a \kl{representation} $g = \repr{R}{t}$ of degree at most $k$, where
  \[
    \widehat{R}(S) = \sum \Big\{ \widehat{Q}(T) : T \subseteq [n] \text{ and } \pi(T) = S \Big \}
  \]
  for every $S \subseteq [m]$. It is easy to see that since $Q$ is \kl{unbiased}, \kl{positive} and \kl{normalized}, then so must be $R$.
  
  We show that $R$ is regular with positive probability, i.e. that $|\widehat{R}(S) - \widehat{R}(T)| \leq \tau$  for every two sets $S, T \subseteq [m]$ of the same size. For every $S \subseteq [m]$, we can treat $\widehat{R}(S)$ as a random variable depending on $(\pi(1), \dots, \pi(n))$. In this context, $\widehat{R}(S) : [m]^n \to \mathbb{R}$ is a function with \kl{bounded differences property} for differences $(\Deg{Q}{1}, \dots, \Deg{Q}{n})$. Indeed, changing the value of $\pi(i)$ can change the value of $\widehat{R}(S)$ by at most the sum of coefficients of $Q$ corresponding to sets including $i$, which is exactly $\Deg{Q}{i}$. Additionaly, we can bound the sum of squares of differences by 
  \[
    \sum_{i=1}^n \big(\Deg{Q}{i}\big)^2 \leq \left( \max_{i \in [n]} \, \Deg{Q}{i} \right)\left( \sum_{i=1}^n \Deg{Q}{i} \right) < \varepsilon \cdot k.
  \]
  From symmetry, for any two sets $S,T \subseteq [m]$ of the same size, we have
  \[
    \EX\left[\widehat{R}(S)\right] = \EX\left[\widehat{R}(T)\right].
  \]
  Therefore, we are allowed to define numbers $\mu_1, \dots, \mu_m \geq 0$, where $\mu_i$ is the expected value of $\widehat{R}(S)$ for every $i \in [n]$ and set $S \subseteq [n]$ of size $i$. We apply \kl{McDiarmid's inequality} to deduce that the distribution of $\widehat{R}(S)$ is tightly concentrated around its expected value. For every $S \subseteq [n]$ of size $i$, we have
  \[
    \Pr\Big[ \big| \widehat{R}(S) - \mu_i \big| \geq \tau/2 \Big] \leq  2 \exp \left( - \frac{\tau^2/2}{\sum_{j=1}^n \big(\Deg{Q}{j}\big)^2} \right) < 2 \exp \left( - \frac{\tau^2/2}{\varepsilon \cdot k} \right) = 1/2^m,
  \]
  where the last equality follows from plugging the value of $\varepsilon$. Union bound implies that the probability that for every $i \in [m]$ and set $S \subseteq[m]$ of size $i$ we have $\widehat{R}(S) \in (\mu_i - \tau/2, \mu_i + \tau/2)$, is positive. Let $g = \repr{Q}{t}$ be the concrete minor witnessing this. In particular, for every pair of sets $S,T \subseteq [m]$ of the same size, we have 
  \[
    \Big| \widehat{R}(S) - \widehat{R}(T) \Big| \leq \tau.
  \]
  From our choice of $m$ and $\tau$, \cref{lemma:positive_polynomials:similar_edges_give_small_sv} implies that all \kl{Shapley values} of $g$ are bounded by $\delta$. Since $g \in \minion$ and $\delta$ was chosen as an arbitrary positive constant, the proof is finished.
\end{proof}

\subsection{Irregularity implies hardness}

This section is devoted to proving the following proposition, which implies the hardness part of \cref{theorem:dichotomy_for_positive_polynomials} due to \cref{theorem:reductions:layered_condition_implies_hardness}.

\begin{proposition}\label{proposition:positive_polynomials:hardness}
    Suppose $k \geq 1$ is an integer and $\minion \subseteq \PTF_k^+$ is a minion, which does not satisfy the $k$-\kl{regularity condition}. Then $\minion$ satisfies the \kl{layered choice condition}. 
\end{proposition}

The fact that a minion $\minion$ fails to satisfy the \kl{regularity condition} implies that there exists a constant $\varepsilon > 0$, such that every \kl{representation} of every function in $\minion$ has a coordinate with \kl{weight} at least $\varepsilon$. Since the sum of weights in a representation of bounded degree is constant, this suggests a natural choice function for $\minion$: choose the set of coordinates with weight at least $\varepsilon$. While the general direction is correct, we have to deal with the fact that representations are ambiguous. Every function in $\minion$ has an infinite number of (\kl{normalized}) representations of bounded degree. Therefore, the fact that a coordinate has significant weight in one representation does not necessarily mean that there is no other representation in which its weight is negligible. This issue motivates the following definition of \textit{heavy sets}, which are designed to capture significant weight across all possible representations.

\begin{definition}[\intro{Heavy sets}]
    Suppose $k \geq 1$ is an integer, $\alpha > 0$ and $f$ is an $n$-ary function with a \kl{representation} of degree at most $k$. We say that $A \subseteq [n]$ is a $(k,\alpha)$-\textit{heavy set} of $f$ if
    \[
        A \cap \Big\{ i \in [n] : \Deg{Q}{i} \geq \alpha \Big\} \neq \emptyset,
    \]
    for every \kl{representation} $f = \repr{Q}{t}$ of degree at most $k$.
\end{definition}

Obviously, if $\minion$ fails to satisfy the $k$-\kl{regularity condition}, then there exists $\varepsilon > 0$ such that every function in $\minion$ has a $(k, \varepsilon)$-\kl{heavy set} --- the set of all coordinates. However, it turns out that we can do better than that: we can find heavy sets of constant size, i.e. independent of function arity.

\begin{lemma}\label{lemma:positive_polynomials:existence_of_heavy_sets}
    Suppose that $k \geq 1$ and $\minion \subseteq \PTF_k^+$ is a \kl{minion} that does not satisfy the $k$-\kl{regularity condition}, which is witnessed by a constant $\varepsilon > 0$. Then every function in $\minion$ has a $(k, \varepsilon/2)$-\kl{heavy set} of size at most $4k/\varepsilon^2$.
\end{lemma}

\begin{proof}
    Fix any $f \in \minion$. We will construct a $(k, \varepsilon/2)$-\kl{heavy set} of $f$ of size at most $4k/\varepsilon^2$. Let $n$ be the arity of $f$ and $f = \repr{Q_0}{t_0}$ be an arbitary \kl{representation} of degree at most $k$. Additionally, let $A_0$ be the set of coordinates with \kl{weight} at least $\varepsilon/2$ in $Q_0$, i.e. $A_0 = \{i \in [n] : \Deg{Q_0}{i} \geq \varepsilon/2 \}$. Let $r = \lceil 2/\varepsilon \rceil$. We will define at most $r$ sets of coordinates $A_1, \dots, A_r \subseteq [n]$ as follows. Iteratively for $i = 1,2,\dots, r$, we perform the following procedure:
    \begin{enumerate}
        \item If $A_0 \cup \dots \cup A_{i-1}$ is a $(k, \varepsilon/2)$-heavy set of $f$, then finish the procedure. The sets $A_i, \dots, A_r$ are left undefined.
        \item Otherwise, let $f = \repr{Q_i}{t_i}$ be a \kl{representation} of degree at most $k$, such that all coordinates in $A_0 \cup \dots A_{i-1}$ have weight less than $\varepsilon/2$ in $Q_i$. Then let $A_i = \{j \in [n] : \Deg{Q_i}{j} \geq \varepsilon/2 \}$.
    \end{enumerate}
    First, suppose that the last defined set is $A_i$ for some $i < r$. Observe that since the degree of all representations is at most $k$, the sum of weights of coordinates in each of them is also bounded by $k$, and therefore $|A_j| \leq 2k/\varepsilon$ for every $j \in \{0, \dots, i-1 \}$. Therefore the set $A_0 \cup \dots \cup A_{i-1}$ is a $(k, \varepsilon/2)$-heavy set of $f$ of size at most $(r-1) \cdot 2k/\varepsilon \leq 2/\varepsilon \cdot 2k/\varepsilon = 4k/\varepsilon^2$, which finishes the construction.

    It remains to consider the case when the procedure reaches the last iteration without finding any viable $(k, \varepsilon/2)$-heavy set. In fact, we want to show that this cannot happen. Assume to the contrary, that the procedure yields representations $\repr{Q_0}{t_0}, \repr{Q_1}{t_1}, \dots, \repr{Q_r}{t_r}$ and sets $A_0, A_1, \dots, A_r$. Consider the function $\repr{\widetilde{Q}}{\tilde t}$ defined as follows:
    \[
        \widetilde{Q} = \frac{1}{r+1}\left( Q_0 + Q_1 + \dots + Q_r \right), \text{\hspace*{0.7cm}} \tilde{t} = \frac{1}{r+1} \left( t_0 + t_1 + \dots t_r \right).
    \]
    It is easy to see that $\widetilde{Q}$ is \kl{unbiased}, \kl{positive} and \kl{normalized}. Furthermore, $\repr{\widetilde{Q}}{\tilde{t}}$ is another \kl{representation} of $f$, i.e. $f = \repr{\widetilde{Q}}{\tilde{t}}$.

    Fix any $j \in [n]$. We observe, that among all representations $\repr{Q_0}{t_0}, \dots, \repr{Q_r}{t_r}$, the weight of coordinate $j$ is less than $\varepsilon/2$ in all but at most one of them, in which it is at most $1$. Therefore, we have the following bound for weight of coordinate $j$ in $\widetilde{Q}$:
    \[
        \omega_j \Big[ \widetilde{Q} \Big] = \frac{1}{r+1} \left( \sum_{i=0}^r \sum_{S \ni j} \widehat{Q_i}(S) \right) = \frac{1}{r+1} \left( \sum_{i=0}^r \Deg{Q_i}{j} \right) \leq \frac{1 + r \cdot \varepsilon/2}{r+1} < \frac{1}{r} + \frac{\varepsilon}{2} \leq \varepsilon.
    \]
    Since $j \in [n]$ was chosen arbitrarily, this must hold for every coordinate. Therefore, we found a \kl{representation} of $f$ of degree at most $k$, in which every coordinate has weight less than $\varepsilon$. This is a direct contradiction with the assumption that $\minion$ fails to satisfy the $k$-\kl{regularity condition} with $\varepsilon$, and thus the procedure must terminate by finding a suitable $(k,\varepsilon/2)$-heavy set of $f$.
\end{proof}

We are now ready to prove the main result of this section. It turns out that choosing the heavy sets yields a proper choice function, but we have to argue that significant \kl{weight} is properly trasfered through minors maps. To this end, we consider representations of minors induced by minor maps, in the spirit of the proof of \cite[Theorem 3.2]{injective_hardness_condition}.

\begin{proof}[Proof of \cref{proposition:positive_polynomials:hardness}]
    Let $k \geq 1$ and $\minion \subseteq \PTF_k^+$ be a minion that does not satisfy the $k$-\kl{regularity condition}, which is witnessed by a constant $\varepsilon > 0$. We want to show that $\minion$ satisfies the \kl{layered choice condition}.

    Let $M = \lceil 2k/\varepsilon +1 \rceil$ and $C$ be a \kl{choice function} for $\minion$ such that $C(f)$ is a $(k, \varepsilon/2)$-\kl{heavy set} of $f$ of size at most $4k/\varepsilon^2$, whose existence is asserted by \cref{lemma:positive_polynomials:existence_of_heavy_sets}. Clearly, $|C(f)| \leq M$ for every $f \in \minion$. It remains to show that $C$ is compatible with chains of \kl{minor maps}. Assume to the contrary that it is not: let $f_1 \xrightarrow{\pi_{1,2}} f_2 \xrightarrow{\pi_{2,3}} \dots \xrightarrow{\pi_{M-1, M}}f_M$ be a chain of minors in $\minion$, such that for every $1 \leq i < j \leq M$ we have
    \begin{equation*}
        \pi_{i, j}(C(f_i)) \cap C(f_j) = \emptyset. \tag{{\color{magenta}$*$}}\label{eq:positive_polynomials:minor_map_disjointness}
    \end{equation*}
    Let $n_1, \dots, n_M$ be arities of functions $f_1, \dots, f_M$. Moving forward, let $f_1 = \repr{Q_1}{t}$ be a \kl{representation} of degree at most $k$ and let $f_2 = \repr{Q_2}{t}, \dots, f_M = \repr{Q_M}{t}$ be representations induced by $\repr{Q_1}{t}$ and minor maps $\pi_{1,2}, \pi_{2,3}, \dots, \pi_{M-1, M}$. In other words, for every $i \in \{0, \dots, M - 1 \}$ and $S \subseteq [n_{i+1}]$ we have
    \[
        \widehat{Q_{i+1}}(S) = \sum \Big\{ \widehat{Q}_i(T) : T \subseteq [n_i] \text{ and } \pi_{i, i+1}(T) = S \Big\}.
    \]
    For every $j \in [M]$, let $B_j \subseteq [n_1]$ be the pre-image of $\pi_{1, j}$ on $C(f_j)$, that is, $B_j = \pi_{1,j}^{-1}(C(f_j))$. Observe that $B_1, \dots, B_M$ are pairwise disjoint; suppose, to the contrary, that $\ell \in B_i \cap B_j$ for $i < j$, we then have $\ell \in C(f_i)$ and $\pi_{i, j}(\ell) \in C(f_j)$, contradicting \eqref{eq:positive_polynomials:minor_map_disjointness}. Furthermore, for every $j \in [M]$, since $C(f_j)$ is a heavy set, we obtain that the sum of \kl{weights} of coordinates in $B_j$ is large:
    \begin{align*}
        \sum_{\ell \in B_j} \Deg{Q_1}{\ell} \geq \sum \Big\{ \widehat{Q_1}(S) : S \cap B_j \neq \emptyset \Big\} \geq \max \Big\{ \Deg{Q_j}{\ell} : \ell \in C(f_j) \Big\} \geq \varepsilon/2.
    \end{align*}
    Summing over $j \in [M]$, we obtain that the sum of coordinate \kl{weights} in $Q_1$ is at least $M \cdot \varepsilon/2 > k$. This is a contradiction, because $Q_1$ is \kl{normalized}. Therefore, $\minion$ satisfies the \kl{layered choice condition} with constant $M$ and choice function $C$.
\end{proof}

\section{Influence and random 2-to-1 minors}\label{sec:influence}
We begin this section with an introduction of basic notions from \kl{Fourier analysis of Boolean functions}. Most of them are standard; we refer the reader to \cite{O’Donnell_2014} for a more comprehensive treatment of this topic.\footnote{Boolean functions in \cite{O’Donnell_2014} are of the form $\{-1,1\}^n \to \{-1,1\}$, in contrast to our functions of the form $\{0,1\}^n \to \{0,1\}$. The difference is mostly cosmetic; our presentation is similar to e.g. \cite{hypergraphremovallemmas}.} 

Let us remind and expand on notions used in previous sections. If $\tuple x \in \Bool^n$, then $\tuple x \oplus i$ is the tuple $\tuple x$ with the $i$-th entry flipped. We also write $\tuple x ^{\, i \to 0}$ and $\tuple x^{\, i \to 1}$ to denote tuples equal to $\tuple x$ with the value of $x_i$ set to $0$ or $1$, respectively. Furthermore, $\tuple x0, \tuple x1 \in \{0,1\}^n$ are tuples equal to $\tuple x$ with $0$ or $1$ appended. Given a distribution $D$ on $\Bool^n$, we write $\tuple x \sim D$ to denote that $\tuple x$ is chosen from $D$ and for a specific $\tuple x \in \{0,1\}^n$, by $D(\tuple x)$ we denote the measure of $\tuple x$ in distribution $D$. If $A \subseteq \{0,1\}^n$, by $D(A)$ we denote the measure of set $A$, i.e. the sum of measures of elements of $A$. Given a function $f : \Bool^n \to \mathbb{R}$, we denote by $\EX_D[f(\tuple x)]$ or $\EX_D[f]$ the expected value of $f(\tuple x)$ when $\tuple x \sim D$. We mainly focus on \emph{product distributions}. Let $\cube{p}$ denote the $p$-\intro{biased distribution} over $\Bool$, i.e. the probability of $1$ is $p$, and the probability of $0$ is $1-p$. Then $\cube{p, n}$ is the product $p$-biased distribution over $\{0,1\}^n$, where each bit is drawn independently from $\cube{p}$. For conciseness, we sometimes write $\cube{p}$ instead of $\cube{p, n}$ if the dimension is clear from the context. Moreover, we use $\EX_p[\cdot]$ and $\Pr_p[\cdot]$ as shorthands for $\EX_{\tuple x \sim \cube{p}}[\cdot]$ and $\Pr_{\tuple x \sim \cube{p}}[\cdot]$. 

Suppose that $p \in (0,1)$. We consider $L^2(\Bool^n, \cube{p})$, the Hilbert space of functions $f : \{0,1\}^n \to \mathbb{R}$ equipped with inner product $\langle f, g \rangle = \EX_p[f(\tuple x) \cdot g(\tuple x)]$. The norm is defined in standard way $\norm{f} = \sqrt{\langle f, f\rangle}$. We distinguish a set of functions in $L^2(\{ 0,1\}^n, \cube{p})$ called \emph{Fourier characters}, defined as follows.

\begin{definition}[\intro{Fourier characters}]
    Let $p \in (0, 1)$ and $i \in [n]$. The \kl{Fourier character} corresponding to singleton $\{i\}$ is the function $\chi_i \in L^2(\Bool^n, \cube{p})$ defined as 
    \[
        \chi_i(\tuple x) = \frac{x_i - p}{\sqrt{p(1-p)}}.
    \]
    \noindent
    In general, the \kl{Fourier character} corresponding to set $S$ is $\chi_S = \prod_{i \in S} \chi_i$, for every $S \subseteq [n]$.
\end{definition}
\AP
It is well known that \kl{Fourier characters} form an orthonormal basis of $L^2(\Bool^n, \cube{p})$. Therefore, every function $f : \Bool^n \to \mathbb{R}$, seen as an element of $L^2(\{0,1\}^n, \cube{p})$, has a unique representation $f = \sum_{S \subseteq[n]} \hat{f}(S) \chi_S$, where $\hat{f}(S) = \langle f, \chi_S \rangle$. The representation is called the \intro{Fourier decomposition}, while the collection $\{ \hat{f}(S) : S \subseteq [n]\}$ is referred as \intro{Fourier coefficients}. Observe that $\EX_p[\chi_S(\tuple x)] = 0$ for every non-empty $S$, which implies that $\EX_p[f] = \hat{f}(\emptyset)$. Orthonormality of \kl{Fourier characters} gives the classic \intro{Parseval-Plancherel identity}:
\begin{gather*}
    \langle f, g \rangle = \sum_{S \subseteq [n]} \hat{f}(S) \cdot \hat{g}(S) \\ 
    \norm{f}^2 = \sum_{S \subseteq [n]} \hat{f}(S)^2
\end{gather*}
We will also encounter the low or high \intro{degree truncations}. Given a function $f \in L^2(\Bool^n, \cube{p})$ and an integer $d$, by $f^{\leq d} = \sum_{|S| \leq d} \hat{f}(S) \cdot \chi_S$ we denote its \kl{low-degree truncation} to degree $d$. Similarly by $f^{>d}$ we denote its \kl{high-degree truncation}, defined analogously. Moving forward, we recall the notion of \kl{influence} introduced in \cref{sec:introduction}: 

\definitionofinfluence*

We can see that the influence of a coordinate is a measure of coordinate \emph{significance} of given function. In section \cref{sec:positive_polynomials} we utilized another measure of coordinate significance: the \kl{Shapley values}. It turns out that if $f$ is monotone, the \kl{Shapley value} of a coordinate is the average of influence over all $p$-biased distributions (see \cite[Exercise 8.31]{O’Donnell_2014}).

If the function $f$ is Boolean valued, the influence of coordinate $i$ can be interpreted as the probability that flipping the value of $i$-th coordinate changes the function value (multiplied by a term depending on $p$). Additionally, the influence of coordinates of function $f$ has the nice property of being tightly correlated with the \kl{Fourier decomposition} of $f$.

\begin{proposition}\label{proposition:influence:alternative_definitions_of_inf}
    Suppose that $p \in (0,1)$ and $f: \Bool^n \to \Bool$ is a function. For every $i \in [n]$ we have
    \[
        \Inf[(p)]{f, i} \xlongequal{\!(1)\!} \sum_{S \ni i} \hat{f}(S)^2 \xlongequal{\!(2)\!} p(1-p) \cdot \EX_{p} \Big[ \big( f(\tuple x) - f(\tuple x \oplus i) \big)^2 \Big] \xlongequal{\!(3)\!} p(1-p) \cdot \Pr_p \Big[ f(\tuple x) \neq f(\tuple x \oplus i) \Big].
    \]
\end{proposition}

\begin{proof}
    Let $g : \{0,1\}^n\to \mathbb{R}$ be defined as $g(\tuple x) = f(\tuple x) - \EX_{s \sim \cube{p}}[f(\tuple x^{\, i \to s})]$. Clearly, we have $\Inf[(p)]{f, i}~=~\norm{g}^2$. We want to find the \kl{Fourier decomposition} of $g$. From the fact that $\EX_p[\chi_i] = 0$ we obtain that
    \[
        \EX_{s \sim \cube{p}} \Big[ f(\tuple x^{\, i \to s}) \Big] = \sum_{S \subseteq [n]} \hat{f}(S) \cdot \EX_{s \sim \cube{p}} \Big[ \chi_S(\tuple x^{\, i \to s}) \Big] = \sum_{S \not \ni i} \hat{f}(S) \cdot \chi_S(\tuple x).
    \]
    Therefore we have $g = \sum_{S \ni i} \hat{f}(S) \cdot \chi_S$ and the \kl{Parseval-Plancherel} identity implies (1). To obtain (2), consider $\Delta(\tuple x) = f(\tuple x^{\, i \to 1}) - f(\tuple x^{\, i \to 0})$. Observe that $\Delta(\tuple x)$ does not depend on $x_i$. It follows that
    \[
        g(\tuple x) = \begin{cases}
            - p \cdot \Delta(\tuple x) & \text{ if } x_i = 0, \\
            (1-p) \cdot \Delta(\tuple x) & \text{ if } x_i = 1.      
        \end{cases}
    \]
    We can now rewrite norm of $g$ in the following way:
    \begin{align*}
        \norm{g}^2 &= (1-p) \cdot \EX_p \Big[ p^2 \cdot \Delta(\tuple x)^2 \, \Big | \, x_i = 0 \Big] + p \cdot \EX_p \Big[ (1-p)^2 \cdot \Delta(\tuple x)^2 \, \Big | \, x_i = 1 \Big] \\ 
        &= \big[(1-p) p^2  + (1-p)^2 p \big] \cdot \EX_p\Big [  \Delta(\tuple x)^2 \Big] \\
        &= p(1-p) \cdot \EX_p \Big[ \big( f(\tuple x) - f(\tuple x \oplus i) \big)^2 \Big],
    \end{align*}
    which gives (2). The equality (3) follows simply from the fact that $f$ is Boolean valued. 
\end{proof}

In particular, the equality (1) in  \cref{proposition:influence:alternative_definitions_of_inf} implies the following connection between \kl{total influence} and \kl{Fourier decomposition}.

\begin{corollary}\label{corollary:influence:total_influence_formula}
    Suppose that $p \in (0,1)$ and $f : \{0,1\}^n \to \mathbb{R}$ is a function. Then the \kl{total influence} of $f$ can be expressed as
    \[
        \I[(p)]{f} = \sum_{S \subseteq [n]} |S| \cdot \hat{f}(S)^2.
    \]
\end{corollary}

\subsection{Random minors and the pull-back distribution}

Our ultimate goal in this section is to understand the relation between coordinate \kl{influence} in a function and its random minor in attempt to construct hardness reductions from \kl{Label Cover} problems to $\PCSP$s. Clearly, this relation heavily depends on how we draw the random minor map. There are several possibilities; we could draw from a uniform distribution over the space of all minors. This is the most general approach and therefore difficult to analyze. On the other hand, in \cref{sec:positive_polynomials} we were looking for minors satisfying certain properties by drawing a random minor uniformly from all possibilities of a fixed arity. Despite how this setting proved useful in the construction of a specific minor, the fixed arity renders it useless in the context of reductions.

The middle ground between these two is drawing from a uniform distribution over all $d$-to-$1$ minors. This approach imposes structural limitations over the result, while retaining enough strength for \kl{Label Cover} reductions to be feasible, as we have seen in \cref{sec:random_condition} in example of the \kl{random 2-to-1 condition}. Analysis of influence in random $2$-to-$1$ minor was pioneered in \cite{braverman_et_al:LIPIcs.ITCS.2021.27} in order to prove the equivalence of \kl{Unique Games Conjecture} and \kl{Rich 2-to-1 Conjecture} (with imperfect completeness). They showed that if a function over unbiased cube has its Fourier coefficients concentrated on \kl{low-degree} coefficients and its random $2$-to-$1$ minor has a coordinate with significant \kl{influence}, then it almost certainly must have been obtained by identification of a coordinate with significant influence in the original function:

\begin{lemma}[Lemma 23 in \cite{braverman_et_al:LIPIcs.ITCS.2021.27} for Boolean functions]\label{lemma:influence:lemma23}
    Suppose $\delta,\zeta > 0$ and $d \geq 1$ is an integer. There are positive constants $\gamma = \gamma(\delta, \zeta)$ and $\tau = \tau(d, \delta, \zeta)$ such that the following holds. Suppose $f : \{0,1\}^{2n} \to [0,1]$ is a function such that $||f^{>d}||^2 \leq \gamma$. Then
    \[
        \Pr_{\pi} \bigg[ \exists \, j \in [n] : \Inf[(1/2)]{f^\pi, j} \geq \delta \text{ and } \max_{i \in \pi^{-1}(j)} \Inf[(1/2)]{f^{\leq d}, i} \leq \tau \bigg] \leq \zeta,
    \]
    where $\pi: [2n] \to [n]$ is a uniformly random $2$-to-$1$ minor map. 
\end{lemma}

The proof of \cref{lemma:influence:lemma23} depends on deep analytical results for functions on symmetric groups. The brief idea behind it is that the space of $2$-to-$1$ minor maps can be embedded into the symmetric group of $2n$ elements, and therefore the function of influence in the minor can be viewed as a function on symmetric group. Using hypercontractive results of \cite{Filmus_Kindler_Lifshitz_Minzer_2024}, it is then possible to bound higher moments of this function in order to show a tight concentration around the expectation. After that, the result follows from a union-bound argument.

We propose a different result, which states that significant coordinate influence over the $p$-\kl{biased distribution} (with $p$ bounded away from $0$ and $1$) is preserved with constant probability, as long as the total influence of the original function is not too large.

\influencepreservation*

Compared to \cref{lemma:influence:lemma23}, our result has two drawbacks. First, and most notably, \cref{lemma:influence:lemma23} states that, with high probability, all influential coordinates of the minor must have been obtained from influential coordinates in the original function, but our result allows us to focus only on one coordinate with no control over the rest at all. Secondly, \cref{proposition:influence:our_main_result} only asserts a constant (but possibly very small) probability that the desirable event will occur, while \cref{lemma:influence:lemma23} states that the event is almost certain. However, these two issues do not pose any obstacle in our hardness reduction applications, as we will see later in the proof of \cref{theorem:hardness_for_ptfs_with_large_influences}.

Furthermore, the proof of \cref{proposition:influence:our_main_result} is elementary and avoids any hypercontractivity or analysis of functions on symmetric groups, and generalizes the assumptions of \cref{lemma:influence:lemma23} in two ways: (1) the setting of unbiased distribution is generalized to $p$-biased distributions, and (2) the assumption of bounded total influence is a direct weakening of the assumption of \kl{low-degree concentration}, as the \kl{total influence} of a function $f: \{0,1\}^{2n} \to [0,1]$ satisfying $||f^{>d}||^2 \leq \gamma$ is at most $2n\gamma + d$ due to \cref{corollary:influence:total_influence_formula}.

In order to analyze the influence of a random minor, we utilize the notion of \textit{pull-back distribution}  (originally introduced in \cite{braverman_et_al:LIPIcs.ITCS.2021.27}), which is the distribution over $\{0,1\}^{2n}$ obtained by drawing a random $2$-to-$1$ map $\pi$, an element $\tuple x$ of $\{0,1\}^n$ over a $p$-biased distribution and then ``pulling'' $\tuple x$ back through $\pi^{-1}$.

\begin{definition}[\intro{Pull-back distribution}]
    Suppose $p \in (0,1)$ and $n \in \mathbb{N}_+$. The \textit{pull-back distribution} $\pullback{p, 2n}$ over $\{0,1\}^{2n}$ is defined by the following process: draw a uniformly random map $\pi : [2n] \to [n]$, a tuple $\tuple x \sim \cube{p, n}$ and return $\tuple z = \pi^{-1}(\tuple x)$.
\end{definition}

Similarly to biased distributions, we sometimes write $\pullback{p}$ instead of $\pullback{p, 2n}$, if the dimension is clear from the context. The usefulness of \kl{pull-back distribution} arises from its clear connection with the process of taking a random minor, and then evaluating it in $p$-biased distribution. We want to analyze how similar $\pullback{p, 2n}$ and $\cube{p, 2n}$ are. It is clear that they are not the same: the support of the former consists only of elements with an even number of ones (and an even number of zeros). However, it turns out that this is the only case where an element $\tuple z \in \{0,1\}^{2n}$ has a smaller measure in the pull-back distribution than in the classic biased distribution.

\begin{lemma}\label{lemma:influence:density_larger_in_pullback}
    There exists a universal constant $C > 0$ such that for all $p \in (0,1)$, $n \in \mathbb{N}_+$ and $\,\tuple z \in \{0,1\}^{2n}$ with an even number of ones, we have
    \[
        \pullback{p}(\tuple z) \geq C \cdot \cube{p}(\tuple z).
    \]
\end{lemma}

\begin{proof}
    Fix numbers $p \in (0,1)$, $n \in \mathbb{N}_+$ and a tuple $\tuple z \in \{0,1\}^{2n}$ with $2k$ ones for some $k \in \{0, 1, \dots, n\}$. First, we calculate the probability that for a uniformly chosen 2-to-1 map $\pi : [2n] \to [n]$, the tuple $\tuple z$ is consistent with $\pi$, i.e. that $\tuple z$ is in the image of $\tuple x \mapsto \pi^{-1}(\tuple x)$ (in other words, that $\pi$ pairs ones with ones and zeros with zeros in $\tuple z$). Let $\mathcal{S}$ be the set of all 2-to-1 maps. Every map is consistent with ${n \choose k}$ tuples, and hence there are $|\mathcal{S}| \cdot {n \choose k}$ pairs $(\pi, \tuple z)$ such that $\tuple z$ is consistent with $\pi$. By symmetry, every tuple is consistent with the same number of maps, equal to $|\mathcal{S}| \cdot {n \choose k} / {2n \choose 2k}$. Dividing by the number of all maps, we obtain:
    \[
        \forall \, \tuple z \in \{0,1\}^{2n} : \Pr_\pi\big[\tuple z \, \text{ is consistent with } \pi\big] = \frac{{n \choose k}}{{2n \choose 2k}}.
    \]
    We obtain the equality $\pullback{p}(\tuple z) = p^k \cdot (1-p)^{n-k} \cdot {n \choose k}/{2n \choose 2k}$. If $k = 0$ or $k = n$, then the lemma statement clearly holds for $C = 1$. Otherwise, we use Stirling's approximation $n! \sim \sqrt{2\pi n}(n/e)^n$ to bound the ratio of binomial coefficients:
    \begin{align*}
        \frac{{n \choose k}}{{2n \choose 2k}} = \frac{n! \cdot (2k)! \cdot (2n-2k)!}{k! \cdot (n-k)! \cdot (2n)!} \sim \frac{\sqrt{4nk(n-k)} \cdot n^n \cdot (2k)^{2k} \cdot (2n-2k)^{2n-2k}}{\sqrt{2nk(n-k)} \cdot k^k \cdot (n-k)^{n-k} \cdot (2n)^{2n}} \geq C \cdot \frac{k^k \cdot (n-k)^{n-k}}{n^n},
    \end{align*}
    where $C$ is a universal positive constant. Moving forward, we want to show that 
    \[
        \frac{k^k \cdot (n-k)^{n-k}}{n^n} \geq p^k \cdot (1-p)^{n-k}.  \tag{{\color{magenta}$*$}}\label{eq:influence:stirling_approximation}
    \]
    To this end, we consider the derivative of $\zeta: p \mapsto p^k \cdot (1-p)^{n-k}$ for $p \in (0,1)$. Simple calculations show that $\zeta'$ is positive in the range $(0, k/n)$, negative in the range $(k/n, 1)$, and equal to $0$ in $k/n$. Therefore, $\zeta$ attains its maximum in the range $(0,1)$ for $p = k/n$, in which case \eqref{eq:influence:stirling_approximation} becomes an equality. Finally, we obtain the desired inequality:
    \[
        \pullback{p}(\tuple z) = p^k \cdot (1-p)^{n-k} \cdot \frac{{n \choose k}}{{2n \choose 2k}} \geq C \cdot  p^{2k} \cdot (1-p)^{2n-2k} = C \cdot \cube{p}(\tuple z). \qedhere
    \]
\end{proof}

We note that the inverse of \cref{lemma:influence:density_larger_in_pullback} (with a different constant) is known to be true if $p = 1/2$ and $\tuple x$ is \textit{roughly balanced} (see \cite[Lemma 27]{braverman_et_al:LIPIcs.ITCS.2021.27}). In fact, we believe that it can be generalized to any $p \in (0,1)$ with an appropriate definition of roughly balanced inputs; we leave it as an exercise for the curious reader.

\subsection{Proof of \cref{proposition:influence:our_main_result}}

This subsection is devoted to the proof of \cref{proposition:influence:our_main_result}. We start with the following lemma, which states that identification of two coordinates $i, j \in [m]$ of an $m$-ary function produces a minor in which the \kl{influence} of the obtained coordinate is controlled. For a simpler presentation, we assume that indices of the identified coordinates are $m-1$ and $m$, but this clearly is without loss of generality.

\begin{lemma}\label{lemma:influence:gluing_first_pair}
    Suppose that $p \in (0,1)$ and $f : \{0,1\}^m \to \{0,1\}$ is a Boolean function. Let $\pi : [m] \to [m-1]$ be a minor map such that $\pi(m-1) = \pi(m) = m-1$ and $\pi(i) = i$ for every $i \in [m-2]$. Then 
    \[
        \Inf[(p)]{f^\pi, m-1} \geq \min(p,1-p) \cdot \Inf[(p)]{f, m-1} - \frac{1}{\min(p, 1-p)} \cdot \Inf[(p)]{f, m}.
    \]
\end{lemma}

\begin{proof}
    Fix $p \in (0,1)$ and a function $f : \{0,1\}^m \to \{0,1\}$. Let $\pi$ be the minor map as in the lemma statement. We assume that $p \leq 1/2$ --- the opposite case is symmetric. We distinguish several subsets of the space of all inputs:
    \begin{gather*}
        \forall \, i \in \{m-1, m\} : \mathbf{Piv}(i) = \Big\{ \tuple x \in \{0,1\}^m : f(\tuple x) \neq f(\tuple x \oplus i) \Big\} \\ 
        \mathbf{Same} = \Big\{ \tuple x \in \{0,1\}^m : x_{m-1} = x_m \Big\} \\ 
        \mathbf{Diff} = \{0,1\}^m \setminus \mathbf{Same}
    \end{gather*}
    The lemma will follow from several claims regarding measures of these sets. \cref{proposition:influence:alternative_definitions_of_inf} yields the first equality:
    \[
        \cube{p}(\mathbf{Piv}(m)) = \frac{1}{p(1-p)} \cdot \Inf[(p)]{f, m}. \tag{{\color{magenta}$*$}} \label{eq:influence:b_is_small}
    \]
    Next, we want to show that $\mathbf{Piv}(m-1) \cap \mathbf{Same}$ is a significant part of $\mathbf{Piv}(m-1)$, in terms of measure. Observe that if $\tuple x \in \mathbf{Piv}(m-1)$, then $(x \oplus m-1) \in \mathbf{Piv}(m-1)$ too, and exactly one element of $\{\tuple x, \tuple x \oplus m-1 \}$ is in $\mathbf{Same}$, while the other is in $\mathbf{Diff}$. Using $\cube{p}(\tuple x \oplus m-1) \geq (p/1-p) \cdot \cube{p}(\tuple x)$ (because $p \leq 1/2$), we obtain that
    \[
        \cube{p}(\mathbf{Piv}(m-1) \cap \mathbf{Same}) \geq \frac{p}{1-p} \cdot \cube{p}(\mathbf{Piv}(m-1) \cap \mathbf{Diff}).
    \]
    From the fact that $\mathbf{Same} \sqcup \mathbf{Diff} = \{0,1\}^m$, it further follows that 
    \[
        \cube{p}(\mathbf{Piv}(m-1)) = \cube{p}(\mathbf{Piv}(m-1) \cap \mathbf{Same}) + \cube{p}(\mathbf{Piv}(m-1) \cap \mathbf{Diff}) \leq \frac{1}{p} \cdot \cube{p}(\mathbf{Piv}(m-1) \cap \mathbf{Same}).
    \]
    The above inequality combined with \cref{proposition:influence:alternative_definitions_of_inf} yields the following bound:
    \[
        \cube{p}(\mathbf{Piv}(m-1) \cap \mathbf{Same}) \geq p \cdot \cube{p}(\mathbf{Piv}(m-1)) \geq \Inf[(p)]{f, m-1}.  \tag{{\color{magenta}$**$}}\label{eq:influence:a_cap_s_is_large}
    \]
    Let $\beta_0$ and $\beta_1$ be respective distributions over $\{0,1\}^m$ of $\pi^{-1}(\tuple x)$ and $(\pi^{-1}(\tuple x) \oplus m-1)$, given that $\tuple x \sim~\cube{p, m-1}$. We can bound the influence of the coordinate $m-1$ in $f^\pi$ as follows:
    \begin{align*}
        \Inf[(p)]{f^\pi, m-1} &= p(1-p) \cdot \Pr_{\tuple z \sim \beta_0} \Big[ f(\tuple z) \neq f(\tuple z \oplus m-1 \oplus m) \Big]
         \\
         &\geq p(1-p) \cdot \Pr_{\tuple z \sim \beta_0} \Big[ f(\tuple z) \neq f(\tuple z \oplus m-1) \text{ and } f(\tuple z \oplus m-1) = f(\tuple z \oplus m-1 \oplus m) \Big].
    \end{align*}
    Using $\Pr[A \cap B] \geq \Pr[A] -\Pr[\neg B]$, we obtain that the value above is at least
    \[
        p(1-p) \cdot \left( \Pr_{\tuple z \sim \beta_0} \Big[ f(\tuple z) \neq f(\tuple z \oplus m-1)  \Big] - \Pr_{\tuple y \sim \beta_1} \Big[ f(\tuple y) \neq f(\tuple y \oplus m) \Big] \right).
    \]
    Observe that (1) if $\tuple x \in \mathbf{Same}$, then $\beta_0(\tuple x) \geq 1/(1-p) \cdot \cube{p}(\tuple x)$, and (2) for every $\tuple x \in \{0,1\}^m$, we have $\beta_1(\tuple x) \leq (1/p) \cdot \cube{p}(\tuple x)$.
    We finish the proof by plugging these two observations along with \eqref{eq:influence:b_is_small} and \eqref{eq:influence:a_cap_s_is_large} into the above inequality:
    \begin{align*}
        \Inf[(p)]{f^\pi, m-1} &\geq p(1-p) \cdot \left( \Pr_{\tuple z \sim \beta_0} \Big[ f(\tuple z) \neq f(\tuple z \oplus m-1)  \Big] - \Pr_{\tuple y \sim \beta_1} \Big[ f(\tuple y) \neq f(\tuple y \oplus m) \Big] \right) \\ 
        &= p(1-p) \cdot \Big( \beta_0(\mathbf{Piv}(m-1) \cap \mathbf{Same}) - \beta_1(\mathbf{Piv}(m)) \Big) \\ 
        &\geq p \cdot \cube{p}(\mathbf{Piv}(m-1) \cap \mathbf{Same}) - (1-p) \cdot \cube{p}(\mathbf{Piv}(m)) \\
        &\geq p \cdot \Inf[(p)]{f, m-1} - \frac{1}{p} \cdot \Inf[(p)]{f, m}. \qedhere
    \end{align*}
\end{proof}

Before we discuss the main proof of this section, we need one more technical lemma. As we shall see later, the \kl{pull-back distribution} is useful for modeling expected values of random $2$-to-$1$ minors. It will be particularly useful for us to deduce that if the expected value of a function is significant in the $p$-biased distribution, then it remains significant over the pull-back distribution. Unfortunately, this is not true for general functions: Consider the function $\mathsf{XOR} : \tuple x \mapsto \sum x_i \, (\text{mod } 2)$ of even arity. It is an easy exercise to show that the expected value of $\mathsf{XOR}$ approaches $1/2$ as its arity grows, over any $p$-biased distribution. However, the expected value of $\mathsf{XOR}$ over the \kl{pull-back distribution} is always $0$; this is because the support of pull-back distribution consists only of points with an even number of ones.

Despite this, it turns out that as long as the \kl{total influence} of the given function is not too large, the situation illustrated by $\mathsf{XOR}$ cannot happen.

\begin{lemma}\label{lemma:influence:l2_norm_remains_large_in_pull_back}
    For every $\lambda \in (0, 1/2)$ and $\delta > 0$, there are positive constants $\gamma = \gamma(\lambda, \delta)$ and $\beta = \beta(\lambda, \delta)$ such that the following holds. Suppose that $p \in (\lambda, 1-\lambda)$ and $h : \{0,1\}^{2n} \to \{0,1\}$ is a function such that $\EX_p[h] \geq \delta$ and $\I[(p)]{h} \leq \gamma \cdot n$. Then
    \[
        \EX_{\tuple z \sim \pullback{p}}\big[h(\tuple z)\big] \geq \beta.
    \]
\end{lemma}

\begin{proof}
    Fix $\lambda \in (0, 1/2)$, $\delta > 0$ and let $C$ be the universal constant from \cref{lemma:influence:density_larger_in_pullback}. We will show that the following constants satisfy the statement conditions:
    \[
        \gamma = \frac{\lambda(1-\lambda) \cdot \delta}{2}, \text{\hspace*{1cm}} \beta = \min \bigg\{ \frac{C}{2}, \frac{C \cdot \lambda}{4(1-\lambda)} \bigg\} \cdot \delta
    \]
    Suppose that $p \in (\lambda, 1-\lambda)$ and $h : \{0,1\}^{2n} \to \{0,1\}$ is a function as in the statement. Let $\mathbf{A} = h^{-1}(1)$ and let $\mathbf{Even} \sqcup \mathbf{Odd} =~\{0,1\}^{2n}$ be the partition of the input space into sets of points with an even and odd number of ones.
    
    We consider two cases. First, suppose that $\cube{p}(\mathbf{A \cap Even}) \geq \delta/2$. In this case, \cref{lemma:influence:density_larger_in_pullback} implies that 
    \[
        \EX_{\tuple z \sim \pullback{p}}\big[ h(\tuple z) \big] = \sum_{\tuple z \in \mathbf{Even}} \pullback{p}(\tuple z) \cdot h(\tuple z) = \pullback{p}(\mathbf{A \cap Even}) \geq C \cdot \cube{p}(\mathbf{A \cap Even}) \geq C \cdot \delta/2 \geq \beta,
    \]
    which proofs the claim. In the remaining case, we have $\cube{p}(\mathbf{A \cap Even}) < \delta/2$. As a consequence, $\cube{p}(\mathbf{A \cap Odd}) \geq \delta/2$.
    Fix any coordinate $i \in [2n]$ such that $\Inf[(p)]{h, i} \leq \gamma/2$, which must exist because of the bound on the total influence of $h$. Let $\mathbf{Piv}(i) \subseteq \{0,1\}^{2n}$ be the set of points $\tuple z$ such that $h(\tuple z) \neq h(\tuple z \oplus i)$. From \cref{proposition:influence:alternative_definitions_of_inf}, we have 
    \[
        \cube{p}(\mathbf{Piv}(i)) = \frac{1}{p(1-p)} \cdot \Inf[(p)]{h, i} \leq \frac{\gamma}{2\lambda(1-\lambda)}.
    \]
    Now observe that for every $\tuple z \in (\mathbf{A \cap Odd}) \setminus \mathbf{Piv}(i)$, the point $(\tuple z \oplus i)$ is in $\mathbf{A \cap Even}$. Therefore, we obtain that
    \begin{align*}
        \cube{p}(\mathbf{A \cap Even}) &\geq \sum \Big\{ \cube{p}(\tuple  z \oplus i) : \tuple z \in (\mathbf{A \cap Odd}) \setminus \mathbf{Piv}(i) \Big \}  \\ 
        &\geq \frac{\lambda}{1-\lambda} \cdot \cube{p}(\mathbf{A \cap Odd} \setminus \mathbf{Piv}(i)) \\ 
        &\geq \frac{\lambda}{1-\lambda} \cdot \left( \frac{\delta}{2} - \frac{\gamma}{2\lambda(1-\lambda)}\right) \\ 
        &= \frac{1}{1-\lambda} \cdot \frac{\delta}{4},
    \end{align*}
    where the last equality follows from plugging the value of $\gamma$. Finally, we apply \cref{lemma:influence:density_larger_in_pullback} one more time to obtain:
    \[
        \EX_{\tuple z \sim \pullback{p}} \big[ h(\tuple z) \big] = \pullback{p}(\mathbf{A \cap Even}) \geq C \cdot \cube{p}(\mathbf{A \cap Even}) \geq \frac{C \cdot \lambda}{4(1-\lambda)} \cdot \delta \geq \beta. \qedhere
    \]
\end{proof}

All that is left to prove \cref{proposition:influence:our_main_result} is to combine the previous two results of this section. The brief idea is that we split the process of drawing the random 2-to-1 minor map into two steps. First, we choose the coordinate with which the influential coordinate is identified. Thanks to \cref{lemma:influence:gluing_first_pair}, we deduce that the obtained coordinate has significant \kl{influence}. In the second step, we randomly identify all the other coordinates. This is where \cref{lemma:influence:l2_norm_remains_large_in_pull_back} will come in handy: as we shall see, the expected influence in the resulting minor can be modeled using the \kl{pull-back distribution}.

\influencepreservation*

\begin{proof}
    Fix any $\lambda \in (0,1/2)$ and $\delta > 0$. Instead of providing explicit values of constants $\gamma, \tau$ and $n_0$, we will show that they can be appropriately chosen during the proof.
    
    Fix $p \in (\lambda, 1-\lambda)$ and let $f$ be as in the proposition statement. Suppose that $i \in [2n]$ is such that $\Inf[(p)]{f, i} \geq \delta$. Without loss of generality, we assume $i = 2n$. Our goal is to show that the influence of coordinate $2n$ is preserved through a random 2-to-1 minor with constant probability, as long as $\gamma$ is sufficiently small and $n_0$ is sufficiently large.

    We divide the process of drawing the 2-to-1 minor map $\pi$ into two steps. First, we uniformly draw $j \in [2n-1]$ and define a minor map $\pi_0 : [2n] \to [2n-1]$ that identifies $2n$ with $j$ and permutes other coordinates without any other identifications. After that, we uniformly draw a 2-to-1 map $\pi_1 : [2n-1]  \to [n]$ that pairs coordinates from $[2n-1] \setminus  \{ \pi_0(2n) \}$. It is easy to see that the distribution of the composition $\pi_1 \circ \pi_0$ is uniform over all maps $\pi: [2n] \to [n]$.

    \begin{figure}[h]
    \centering
    \includegraphics[scale=0.8]{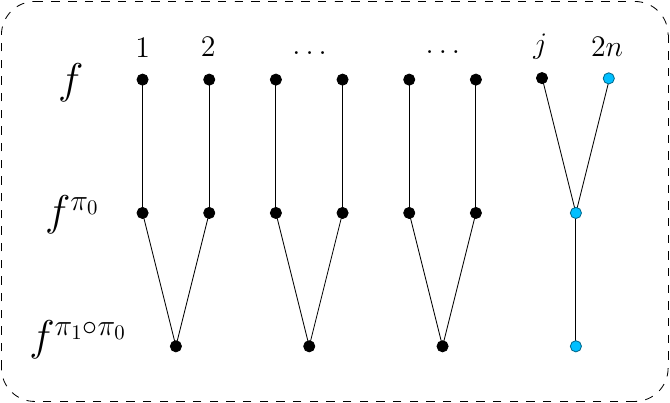}
    \caption{
        We split the random choice of $\pi$ into two steps: (1) choosing the coordinate $j$ with which the influential coordinate $2n$ should be identified, and (2) choosing the pairing of remaining coordinates. We show that after each step, the image of coordinate $2n$ remains influential with constant probability.
    }\label{fig:two_steps_polymorphisms}
    \end{figure}

    Consider the first step: drawing the map $\pi_0$ that identifies $2n$ with a uniformly chosen $j \in [2n-1]$. Let $g = f^{\pi_0}$ and without loss of generality assume that $\pi_0(2n) = 2n-1$. \cref{lemma:influence:gluing_first_pair} implies that
    \begin{align*}
        \EX_{\pi_0} \Big[ \Inf[(p)]{g, 2n-1} \Big] &\geq \lambda \cdot \Inf[(p)]{f, 2n} - \frac{1}{\lambda} \cdot \EX_{j \sim [2n-1]} \Big[ \Inf[(p)]{f, j} \Big] \\
        &= \lambda \cdot \Inf[(p)]{f, 2n} - \frac{1}{\lambda \cdot (2n-1)} \cdot \I[(p)]{f} \\ 
        &\geq \lambda \cdot \delta - \frac{\gamma \cdot n}{\lambda \cdot (2n-1)} \\ 
        &\geq \delta_0(\lambda, \delta),
    \end{align*}
    where $\delta_0 > 0$ if $\gamma$ is sufficiently small and $n_0$ sufficiently large. Using the fact that every random variable $\mathbf{X} \in [0,1]$ satisfying $\EX[\mathbf{X}] \geq \mu$ admits $\Pr[\mathbf{X} \geq \mu/2] \geq \mu/2$, we obtain:
    \[
        \Pr_{\pi_0} \Big[ \Inf[(p)]{g, 2n-1} \geq \delta_0/2 \Big] \geq \delta_0/2. \tag{{\color{magenta}$*$}}\label{eq:influence:main_proof_first_step}
    \]
    Fix any $\pi_0$ witnessing the event in \eqref{eq:influence:main_proof_first_step}. Without loss of generality, assume that $\pi_0$ identifies $2n$ with $2n-1$. Our next goal is to bound the \kl{total influence} of $g$.
    
    \begin{claim}
        $\I[(p)]{g} \leq \gamma_0 \cdot (n-1)$, where $\gamma_0 = \gamma_0(\gamma, n) \to 0$ given $\gamma \to 0$ and $n \to \infty$.
    \end{claim}

    \begin{proof}[Proof of claim]
        Clearly, we have $\Inf[(p)]{g, 2n-1} \leq 1$. Pick any $\ell \in [2n-2]$. Let $\theta$ be the distribution of $\pi_0^{-1}(\tuple x)$, given $\tuple x \sim \cube{p, 2n-1}$. For every $\tuple z \in \{0,1\}^{2n}$ we have $\theta(\tuple z) \leq (1/\min(p, 1-p)) \cdot \cube{p}(\tuple z)$, which is at most $(1/\lambda) \cdot \cube{p}(\tuple z)$. This allows to bound the influence of $\ell$ in $g$ in terms of influence of $\pi_0^{-1}(\ell)$ in $f$ as follows:
        \[
            \Inf[(p)]{g, \ell} = p(1-p) \cdot \Pr_{\tuple z \sim \theta} \Big[ f(\tuple z) \neq f(\tuple z \oplus \pi_0^{-1}(\ell)) \Big] \leq \frac{1}{\lambda} \cdot \Inf[(p)]{f, \pi_0^{-1}(\ell)}.
        \]
        Summing over coordinates, we deduce that 
        \[
            \I[(p)]{g} \leq 1 + \frac{1}{\lambda} \cdot \I[(p)]{f} \leq 1 + \frac{\gamma \cdot n}{\lambda} = (n-1) \cdot \left( \frac{1}{n-1} + \frac{\gamma \cdot n}{(n-1) \cdot \lambda} \right),
        \]
        where the last factor can get arbitrarily close to $0$ with an appropriate choice of $n_0$ and $\gamma$.
    \end{proof}
    
    We can now proceed to the second step: choosing uniformly at random a minor map $\pi_1 : [2n-1] \to [n]$, which pairs coordinates $\{1,\dots, 2n-2\}$. We can assume that $\pi_1(2n-1) = n$. Let $\pi = \pi_1 \circ \pi_0$. We want to understand the expected influence of coordinate $n$ in $f^\pi$. To this end, we analyze the function $h: \{0,1\}^{2n-2} \to \{0,1\}$ defined as follows:
    \[
        h(\tuple x) = \begin{cases}
            1 & \text{ if } g(\tuple x0) \neq g(\tuple x1), \\ 
            0 & \text{ otherwise.}
        \end{cases}
    \]
    By \cref{proposition:influence:alternative_definitions_of_inf}, we have $\EX_p[h] \geq \Inf[(p)]{g, 2n-1} \geq \delta_0/2$. Furthermore, similarly to $g$, the \kl{total influence} of $h$ cannot be too large:
    
    \begin{claim}
        $\I[(p)]{h} \leq \gamma_1 \cdot (n-1)$, where $\gamma_1 = \gamma_1(\gamma_0) \to 0$, given $\gamma_0 \to 0$. 
    \end{claim}

    \begin{proof}[Proof of claim]
        Fix any coordinate $\ell \in [2n-2]$ and let $\mathbf{Piv}(\ell) \subseteq \{0,1\}^{2n-2}$ be the set of points $\tuple x$ such that $h(\tuple x) \neq h(\tuple x \oplus \ell)$. Given $\tuple x \in \mathbf{Piv}(\ell)$, it must hold that $g(\tuple x0) \neq g(\tuple x0 \oplus \ell)$ or $g(\tuple x1) \neq g(\tuple x1 \oplus \ell)$ --- otherwise we would have:
        \[
            h(\tuple x) = 1 \iff g(\tuple x0) \neq g(\tuple x1) \iff g(\tuple x0 \oplus \ell) \neq g(\tuple x1 \oplus \ell) \iff h(\tuple x \oplus \ell) = 1.
        \]
        This allows us to lower bound $\Inf[(p)]{g, \ell}$ in terms of $\Inf[(p)]{h, \ell}$ as follows: 
        \begin{align*}
            \Inf[(p)]{g, \ell} &\geq \frac{1}{2} \cdot p(1-p)\cdot \sum \Big\{ \min_{s \in \{0,1\}} \cube{p}(\tuple x s) : \tuple x \in \mathbf{Piv}(\ell) \Big\} \\ 
            &\geq \frac{1}{2} \cdot p(1-p) \cdot \sum \bigg\{ p \cdot \cube{p}(\tuple x) : \tuple x \in \mathbf{Piv}(\ell) \bigg \} \\ 
            &= \frac{1}{2} \cdot p^2(1-p) \cdot \cube{p}(\mathbf{Piv}(\ell)) \\
            &= (p/2) \cdot \Inf[(p)]{h, \ell},
        \end{align*}
        where the last equality follows from \cref{proposition:influence:alternative_definitions_of_inf}. Summing over coordinates, we have $\I[(p)]{h} \leq (2/p) \cdot \I[(p)]{g} \leq (2/\lambda) \cdot \gamma_0 \cdot (n-1)$.
    \end{proof}
    
    Since $\gamma_1$ can be arbitrarily small with an appropriate choice of $\gamma$ and $n_0$ (depending only on $\lambda$ and $\delta$), we are now in a position to apply \cref{lemma:influence:l2_norm_remains_large_in_pull_back} to $h$. We obtain a positive constant $\beta = \beta(\lambda, \delta)$ such that:
    \[
        \beta \leq \EX_{\tuple z \sim \pullback{p}} \big[ h(\tuple z)\big] = \EX_{\pi_1} \bigg [ \EX_{\tuple x \sim \cube{p}} \Big[ h(\pi_1^{-1}(\tuple x)) \Big] \bigg] = \EX_{\pi_1} \bigg[ \Pr_{\tuple x \sim \cube{p}} \Big[ g(\pi_1^{-1}(\tuple x)0) \neq g( \pi^{-1}(\tuple x)1) \Big] \bigg].
    \]
    In addition, observe that $g(\pi_1^{-1}(\tuple x)s) = g^{\pi_1}(\tuple x s) = f^\pi(\tuple x s)$. It follows that 
    \[
        \beta \leq \EX_{\pi_1} \bigg[ \Pr_{\tuple x \sim \cube{p}} \Big[ f^\pi(\tuple x 0) \neq f^{\pi}( \tuple x 1) \Big] \bigg] = \EX_{\pi_1} \bigg[ \Pr_{\tuple z \sim \cube{p}} \Big[ f^\pi(\tuple z) \neq f^\pi(\tuple z \oplus n) \Big] \bigg] = \frac{1}{p(1-p)} \EX_{\pi_1} \Big[ \Inf[(p)]{f^\pi, n} \Big].
    \]
    Using $\Pr[\mathbf{X} \geq \EX[\mathbf{X}]/2] \geq \EX[\mathbf{X}]/2$ one more time and plugging $n = \pi(2n)$, we obtain
    \[
        \Pr_{\pi_1} \Big[ \Inf[(p)]{f^\pi, \pi(2n)} \geq \frac{1}{2} \cdot p(1-p) \cdot \beta \Big] \geq \frac{1}{2} \cdot p(1-p) \cdot \beta.
        \tag{{\color{magenta}$**$}}\label{eq:influence:main_proof_second_step}
    \]
    Finally, recall that we must also count in the probability that $\pi_0$ satisfies the required properties. Combining \eqref{eq:influence:main_proof_first_step} and \eqref{eq:influence:main_proof_second_step}, we obtain the final estimate:
    \[
        \Pr_{\pi} \Big[ \Inf[(p)]{f^\pi, \pi(2n)} \geq \frac{1}{2} \cdot \lambda(1-\lambda) \cdot \beta \Big] \geq \frac{\delta_0}{4} \cdot \lambda(1-\lambda) \cdot \beta.
    \]
    This finishes the proof because $\delta_0, \beta > 0$ depend only on $\lambda$ and $\delta$.
\end{proof}

\subsection{Hardness for Polynomial Threshold Functions of bounded degree}

In this section, we describe our application of \cref{proposition:influence:our_main_result} for the class of \kl{Polynomial Threshold Functions} by proving \cref{theorem:hardness_for_ptfs_with_large_influences}. The crucial property of $\PTF$s that we are going to utilize is their \intro{low-degree concentration}, which means that they are similar to their \kl{low-degree truncation}. This fact can be deduced from a result obtained in \cite{Harsha2009BoundingTS} regarding the \textit{noise sensitivity} of $\PTF$s. 

To properly introduce this notion, we first have to define the \textit{noise operator}. All of the following definitions are standard; we refer to \cite[Section 2.4]{O’Donnell_2014} for a more involved introduction of these concepts in the setting of unbiased distribution. The generalization to $p$-biased distributions is verbatim.

\begin{definition}[\intro{Noise operator}]
    Given $\delta, p \in [0,1]$ and $\tuple x \in \{0,1\}^n$, the $(p, \delta)$-\intro{noisy distribution} of $\tuple x$ denotes the distribution over elements $\tuple y \in \{0,1\}^n$ defined as follows: for every $i \in [n]$ independently, let 
    \[
        y_i = \begin{cases}
            x_i & \text{ with probability $\delta$}, \\
            \text{drawn from } \cube{p} & \text{ with probability } 1-\delta.
        \end{cases}
    \]
    The $(p, \delta)$-noisy distribution of $\tuple x$ is denoted by $N_{\delta, p}(\tuple x)$.

    The \kl{noise operator} $\T_{\delta, p}$ on space $L^2(\Bool^n, \cube{p})$ is the operator that assings to every function $f$ a corresponding function defined as 
    \[
        \T_{\delta, p} [f](\tuple x) = \EX_{\tuple y \sim N_{\delta, p}(\tuple x)} \big[ f(\tuple y) \big].
    \]
\end{definition}

One can think of the parameter $\delta$ as the level of \textit{correlation} between the original point and its perturbation --- the higher the correlation (the closer $\delta$ is to 1), the more similar to $\tuple x$ we expect $\tuple y \sim N_{\delta, p}(\tuple x)$ to be. This is further reflected in the following fact, which says that the \kl{Fourier decomposition} of $\T_{\delta, p}[f]$ resembles the decomposition of $f$ as long as the noise parameter is not far from $1$.

\begin{proposition}\label{proposition:noise:coefficients_of_noise_operator}
    Suppose $f \in L^2(\Bool^n, \cube{p})$ and $\delta \in [0,1]$. The \kl{Fourier decomposition} of $\, \T_{\delta, p}[f]$ is 
    \[
        \T_{\delta, p}[f] = \sum_{S \subseteq[n]} \delta^{|S|} \hat{f}(S) \cdot \chi_S.
    \]
\end{proposition}

\begin{proof}
    Observe that $\T_{\delta, p}$ is a linear operator. Therefore, it suffices to show that $\T_{\delta, p}[\chi_S] = \delta^{|S|} \cdot~\chi_S$. This holds trivially if $S = \emptyset$. If $S$ is not empty, we obtain:
    \[
        \T_{\delta, p}[\chi_S](\tuple x) = \EX_{\tuple y \sim N_{\delta, p}(\tuple x)} \big[ \chi_S(\tuple x) \big] = \prod_{i \in S} \EX_{y \sim N_{\delta, p}(x_i)}\big[ \chi_i(y) \big] = \prod_{i \in S} \delta \cdot \chi_S(x_i) + (1-\delta) \cdot \EX_p[\chi_S] = \delta^{|S|} \cdot \chi_S. \qedhere
    \]
\end{proof}

In pair with the \kl{noise operator} is the notion of \textit{noise sensitivity}, which is a measure of how prone to perturbation of input values a function is.

\begin{definition}[\intro{Noise sensitivity}]
    Suppose $f \in L^2(\Bool^n, \cube{p})$ and $\delta \in [0,1]$. Let $\tuple x \sim \cube{p, n}$ be a $p$-biased vector and $\tuple y \sim N_{\delta, p}(\tuple x)$. The \kl{noise sensitivity} of $f$ at $\delta$ is defined as
    \[
        \NS_{\delta, p}[f] = \Pr \big[ f(\tuple x) \neq f(\tuple y)].
    \]
\end{definition}

We are now in a position to state a result from \cite{Harsha2009BoundingTS}, which bounds the \kl{noise sensitivity}\footnote{The original statement of \cref{theorem:noise:ptfs_are_noise_stable} concerns a perturbation defined in a bit different way than in our presentation. However, it is easy to convince oneself that the original noise distribution is equivalent to \kl{noisy distribution} with the adjusted parameter $\delta$.} of $\PTF$s of bounded degree.

\begin{theorem}[Theorem 1.3 in \cite{Harsha2009BoundingTS}]\label{theorem:noise:ptfs_are_noise_stable}
    Suppose that $\delta \in [0,1]$ and $f \in \PTF_k$. Let $\tuple x \sim \cube{1/2, n}$ be a uniformly random vector and $\tuple y \sim N_{\delta,1/2}(\tuple x)$. It follows that 
    \[
        \NS_{\delta, 1/2}[f] \leq C(k) \cdot \left(\frac{1}{2} -\frac{1}{2} \cdot \delta\right)^{1/(4k+6)}
    \]
\end{theorem}

The purpose of \cref{theorem:noise:ptfs_are_noise_stable} and definitions introduced above will become apparent in the proof of the following result, which states that \kl{Polynomial Threshold Functions} of bounded degree are \kl{low-degree concentrated}. To show this, we use the standard connection between \kl{noise sensitivity} of $f$ and \kl{Fourier coefficients} of $\T_{\delta, p}[f]$ (see e.g. \cite[Proposition 3.3]{O’Donnell_2014}).

\begin{proposition}\label{proposition:noise:ptfs_are_low_degree_concentrated}
    For every $k \geq 1$ and $\gamma > 0$, there exists a constant $d = d(k, \gamma)$ such that for every $f \in \PTF_k$ it holds that
    \[
        \norm{f^{>d}}^2 \leq \gamma,
    \]
    where the norm is taken over the unbiased distribution $\cube{1/2}$.
\end{proposition}

\begin{proof}
    Fix $k \geq 1$ and $\gamma > 0$. Let $f : \{0,1\}^n \to \{0,1\}$ be a $\PTF$ of degree at most $k$. Observe that \cref{theorem:noise:ptfs_are_noise_stable} implies that $\NS_{\delta,1/2}[f]$ approaches $0$ as $\delta \to 1$. Therefore, we can choose $\delta = \delta(k, \gamma)$ such that $\NS_{\delta, 1/2} \leq \gamma/2$. Let $\tuple x \sim \cube{1/2, n}$ and $\tuple y \sim N_{\delta, 1/2}(\tuple x)$. Observe that
    \[
        \langle f, T_{\delta, p}[f] \rangle = \EX \big[ f(\tuple x) \cdot f(\tuple y) \big] \geq \Pr \big[ f(\tuple x) = 1 \big] - \Pr \big[ f(\tuple x) \neq f(\tuple y) \big] = \norm{f}^2 - \NS_{\delta, 1/2}[f].
    \]
    It follows that $\norm{f}^2 - \langle f, \T_{\delta, p}[f] \rangle \leq \gamma/2$. By \cref{proposition:noise:coefficients_of_noise_operator} and the \kl{Parseval identity}, we obtain that
    \[
        \gamma/2 \geq \norm{f}^2 - \langle f, \T_{\delta, 1/2}[f] \rangle = \sum_{S \subseteq [n]} \hat{f}(S)^2 - \sum_{S \subseteq [n]} \delta^{|S|} \cdot \hat{f}(S)^2 = \sum_{S \subseteq [n]} \hat{f}(S)^2 \left(1 - \delta^{|S|} \right).
    \]
    In particular, for every $d \geq 1$, we have $\norm{f^{>d}}^2 \left(1 - \delta^d \right) \leq \gamma/2$. Let $d = d(k, \gamma)$ be such that $(1-\delta^d) \geq 1/2$. It is easy to verify that $d$ satisfies the required condition:
    \[
        \norm{f^{>d}}^2 \leq \frac{\gamma/2}{1-\delta^d} \leq \gamma. \qedhere
    \]
\end{proof}

Although we defined the notions of \kl{noise operator} and \kl{noise sensitivity} for general $p$-biased distribution, \cref{theorem:noise:ptfs_are_noise_stable} (and \cref{proposition:noise:ptfs_are_low_degree_concentrated} as a consequence) applies only to the unbiased distribution. It is an interesting question whether these results can be generalized to more general setting. We are unsure whether the original proof of \cref{theorem:noise:ptfs_are_noise_stable} simply rewrites itself in the regime of biased distribution. We leave this question for future research, but note that it would automatically strenghten \cref{theorem:hardness_for_ptfs_with_large_influences}: as we will see, \cref{proposition:noise:ptfs_are_low_degree_concentrated} is the only element of the proof that requires uniform distribution.

We are finally in position to construct the hardness reduction proving \cref{theorem:hardness_for_ptfs_with_large_influences}. In fact, we prove that a minion of $\PTF$s of bounded degree with influential coordinates in every function satisfies the \kl{random 2-to-1 condition}. This directly implies \cref{theorem:hardness_for_ptfs_with_large_influences} due to \cref{theorem:reductions:random_condition_implies_hardness}.

With all the tools we have in hand, this turns out to be quite straightforward. The general idea is that \cref{proposition:influence:our_main_result} allows us to define a choice function which is compatible with random $2$-to-$1$ minors, while \cref{proposition:noise:ptfs_are_low_degree_concentrated} asserts that the sets of chosen coordinates have bounded size.

\begin{proposition}
    Suppose that $k \geq 1$ is an integer and $\minion \subseteq \PTF_k$ is a minion. If there exists a constant $\delta > 0$ such that 
    \[
        \max_{i \in [n]} \, \Inf[(1/2)]{f, i} \geq \delta
    \]  
    for every $f \in \minion$ of arity $n$, then $\minion$ satisfies the \kl{random 2-to-1 condition}.
\end{proposition}

\begin{proof}
    Fix $k \geq 1$ and $\delta > 0$. Suppose that $\minion \subseteq \PTF_k$ is a minion such that every function in $\minion$ has coordinate with influence at least $\delta$ over uniform distribution. Our goal is to show that $\minion$ satisfies the \kl{random 2-to-1 condition}.

    Let $\gamma, \tau$ and $n_0$ be constants provided by \cref{proposition:influence:our_main_result} for $\delta$ and $\lambda := 1/3$. We can assume that $\tau \leq \delta$, otherwise we can set $\tau := \delta$. \cref{proposition:noise:ptfs_are_low_degree_concentrated} implies that there exists a constant $d = d(k, \delta)$ such that every function $f \in \minion$ satisfies $\norm{f^{>d}}^2 \leq \gamma/3$. In particular, by \cref{corollary:influence:total_influence_formula} we have $\I[(1/2)]{f} \leq d + (2/3)n \gamma$, which is at most $\gamma \cdot n$, as long as $n \geq n_1$ for some $n_1 = n_1(k, \delta)$. We can assume that $n_0 \geq n_1$, otherwise we set $n_0 := n_1$.

    We are ready to define the \kl{choice function} for $\minion$. After that, we will argue that it satisfies all required conditions. Let $C : \minion \to \mathcal{R}(\mathbb{N})$ be the \kl{choice function} that assings to every $n$-ary function $f$ a subset of coordinates as follows:
    \[
        C(f) = \begin{cases}
            \text{the set of coordinates } i \in [n] \text{ such that } \Inf[(1/2)]{f, i} \geq \tau & \text{ if } n \geq 2n_0, \\ 
            [n] & \text{ if } n < 2n_0.
         \end{cases}
    \]

    \begin{claim}
        $|C(f)| \leq M(k, \delta)$ for every $f \in \minion$.
    \end{claim}

    \begin{proof}[Proof of claim]
        Suppose that $f : \{0,1\}^n \to \{0, 1\}$. If $n < 2n_0$, then $|C(f)|$ is less than $2n_0$, which is a constant depending only on $k$ and $\delta$. Otherwise, $C(f)$ consists of coordinates with \kl{influence} at least $\tau$. We will show that \kl{low-degree} concentration of $\PTF$s implies that there can only be a constant number of such coordinates. Suppose that $i \in C(f)$. \cref{proposition:noise:ptfs_are_low_degree_concentrated} gives that there exists a constant $d = d(k, \delta)$ such that $\norm{f^{>d}}^2 \leq \tau/2$. The \kl{Parseval identity} and \cref{proposition:influence:alternative_definitions_of_inf} give us that 
        \[
            \Inf[(1/2)]{f^{\leq d}, i} = \sum \bigg\{ \hat{f}(S)^2 : S \subseteq [n], i \in S \text{ and } |S| \leq d \bigg\} \geq \Inf[(1/2)]{f, i} - \tau/2 \geq \tau/2. \tag{{\color{magenta}$*$}}\label{eq:noise:constant_low_degree_influence}
        \]
        However, by \cref{corollary:influence:total_influence_formula} we have $\I[(1/2)]{f^{\leq d}} \leq d \cdot \norm{f}^2 \leq d$. Therefore, there can be at most $(2d)/\tau$ coordinates satisfying \eqref{eq:noise:constant_low_degree_influence}. As a consequence, $|C(f)| \leq (2d)/\tau$, which is a constant depending only on $k$ and $\delta$.
    \end{proof}

    The rest is to show that $C$ is compatible with random $2$-to-$1$ minors. This is an immediate consequence of \cref{proposition:influence:our_main_result}. Suppose that $f : \Bool^{2n} \to \Bool$ belongs to $\minion$. If $n < 2n_0$, then both $f$ and all its $2$-to-$1$ minors have all their coordinates chosen by $C$, therefore $\Pr_{\pi}[\pi(C(f)) \cap C(f^\pi) \neq \emptyset] = 1$. Otherwise, we know that $\I[(1/2)]{f} \leq n \cdot \gamma$ and there is a coordinate $i \in [2n]$ with $\Inf[(1/2)]{f, i} \geq \delta$. In particular, we have $i \in C(f)$ because $\tau \leq \delta$. \cref{proposition:influence:our_main_result} implies that $\Inf[(1/2)]{f^\pi, \pi(i)}$ is at least $\tau$ with probability at least $\tau$, which is equivalent to the event that $\pi(i) \in C(f^\pi)$. In particular,
    \[
        \Pr \Big[ \pi(C(f)) \cap C(f^\pi) \neq \emptyset \Big] \geq \tau.
    \]
    This finishes the proof, because $\tau$ is a constant which depends only on $\delta$.
\end{proof}

\section{Conclusions and future research}\label{sec:future}
As we have seen, we were able to identify several conditions that guarantee complexity characterizations of $\PCSP$s with \kl{polymorphism} \kl{minions} consisting of \kl{Polynomial Threshold Functions} of bounded degree. In \cref{sec:positive_polynomials} we showed that if we forbid negative coefficients in $\PTF$ representations, a complete dichotomy is admitted (\cref{theorem:dichotomy_for_positive_polynomials}). This restriction forces all polymorphisms to be \kl{monotone} functions. It seems to be a natural question whether the non-negativity of coefficients is essential in our result, i.e. whether one can strenghten \cref{theorem:dichotomy_for_positive_polynomials} by allowing negative coefficients, while keeping the assumption of monotonicity.

\begin{question}\label{question:monotone_ptfs}
    Does the class of $\,\PCSP$s with \kl{polymorphism} \kl{minions} consisting of \kl{monotone} $\PTF$s of bounded degree admit a complexity dichotomy?
\end{question}

In \cref{sec:influence}, we established a hardness condition for minions consisting of general $\PTF$s of bounded degree, based on the notion of \kl{influence} over uniform distribution (\cref{theorem:hardness_for_ptfs_with_large_influences}). As we have already observed, the only obstacle to generalization of \cref{theorem:hardness_for_ptfs_with_large_influences} to general $p$-biased distributions is the fact that the results of \cite{Harsha2009BoundingTS} have been obtained specifically for the uniform distribution. We therefore state the following question, which might also be of interest for research independent of $\PCSP$s.

\begin{question}\label{question:ptfs_low_degree_concentration_in_general_spaces}
    Is the class of $\, \PTF$s of bounded degree \kl{low-degree concentrated} over general $p$-biased distributions? In other words, can \cref{proposition:noise:ptfs_are_low_degree_concentrated} be extended beyond the uniform distribution?
\end{question}

Another direction in which \cref{theorem:hardness_for_ptfs_with_large_influences} could be extended is in the search for a tractability counterpart. Although we provided a hardness condition, we highly doubt that it captures all hard $\PCSP$s in the setting of $\PTF$s with bounded degree. Looking at the problem from the other side could shed some light on the general understanding of such $\PCSP$s. To our knowledge, the algorithm for \kl{Linear Threshold Functions} of \cite{injective_hardness_condition} and \cref{proposition:positive_polynomials:tractability} are the only tractability results in this direction.

\begin{question}\label{question:algorithms_for_ptfs}
    Is there a polynomial-time algorithm solving $\PCSP$s with \kl{polymorphism} \kl{minions} consisting of $\PTF$s of bounded degree with no significant coordinates?
\end{question}

We note that searching for a tractability result for \cref{question:monotone_ptfs} could potentially be a reasonable first step toward a resolution of \cref{question:algorithms_for_ptfs}.

Finally, we point to the main motivation of the whole line of study of Boolean PCSPs, which was also the main reason for our research to begin with. Namely, we ask whether an extension of Schaefer's Dichotomy of \cite{schaefer} to $\PCSP$s is true.

\begin{question}
    Does the class of Boolean $\PCSP$s admit a complexity dichotomy?
\end{question}

\section*{Acknowledgments}
I want to thank Marcin Kozik for supervision of my Master's degree project, of which this paper is a conclusion. I also thank Demian Banakh for help in extending the notion of pull-back distribution to general $p$-biased distributions, as well as Tomasz Mazur, who was the first to formally wrap up the proof of \cref{proposition:positive_polynomials:tractability}.

\printbibliography

\end{document}